\documentclass[11pt]{article}
\usepackage{graphicx} 
\usepackage{amsthm,amssymb}
\usepackage{complexity}

\usepackage{fullpage}

\let\cP\relax
\let\E\relax
\let\R\relax

\newcommand{\sqrtsum}{\textsc{SqrtSum}}

\usepackage{bm}
\usepackage{fragkia}
\usepackage[colorlinks, citecolor=teal]{hyperref}

\usepackage{tikz}
\usetikzlibrary{trees,positioning,shapes,calc}

\hypersetup{colorlinks, citecolor=magenta, linkcolor=magenta}
\usepackage{xspace}
\usepackage{physics}
\usepackage{algpseudocode}
\usepackage[shortlabels]{enumitem}
\usepackage{bbm}
\usepackage{natbib}
\let\cite\citep
\usepackage[ruled]{algorithm2e} 

\SetAlFnt{\small}
\SetAlCapFnt{\small}
\SetAlCapNameFnt{\small}
\SetAlCapHSkip{0pt}
\IncMargin{-\parindent}

\usepackage{hyperref}
\usepackage[nameinlink]{cleveref}
\hypersetup{colorlinks, citecolor=magenta, linkcolor=magenta}
\newcommand{\tvx}{\tilde{\vx}}
\newcommand{\SW}{\mathsf{SW}}

\newcommand{\delimit}[3]{\newcommand{#1}[1]{\left#2##1\right#3}}
\delimit \ceil \lceil \rceil
\delimit \floor \lfloor \rfloor
\def\va{{\bm{a}}}
\def\vb{{\bm{b}}}
\def\vc{{\bm{c}}}
\def\vd{{\bm{d}}}

\def\vf{{\bm{f}}}

\def\vh{{\bm{h}}}

\def\vo{{\bm{o}}}

\def\vu{{\bm{u}}}
\def\vv{{\bm{v}}}
\def\vw{{\bm{w}}}
\def\vx{{\bm{x}}}
\def\vy{{\bm{y}}}
\def\vz{{\bm{z}}}

\def\mA{{\mathbf{A}}}

\let\eps\epsilon
\newcommand{\commentsymbol}{\it\color{gray}$\triangleright$~}
\algrenewcommand\algorithmiccomment[1]{\hfill{\commentsymbol#1}}

\algdef{SxnE}[Block]{Block}{EndBlock}[1]{{#1}}
\newcommand{\ie}{{\em i.e.}\xspace}
\newcommand{\eg}{{\em e.g.}\xspace}
\let\op\operatorname
\newcommand\size{\op{size}}
\newcommand\round{\op{round}}
\let\Root\varnothing

\let\ip\ev
\newcommand\sharpP{{\sf \#P}}

\usepackage{multirow}

\newcommand{\solcon}{\mathsf{SolCon}}

\usepackage{pifont}
\newcommand{\cmark}{\ding{51}}%
\newcommand{\xmark}{\ding{55}}%
\newcommand{\supp}{\op{supp}}
\newcommand\ind{\mathbbm 1}

\newcommand{\ioannis}[1]{\textcolor{green}{[Ioannis: #1]}}
\definecolor{briancolor}{rgb}{0, .5, 0}
\newcommand{\brian}[1]{\textcolor{briancolor}{[Brian: #1]}}

\renewcommand{\vec}[1]{\bm{#1}}
\newcommand{\mat}[1]{\mathbf{#1}}

\newcommand{\lesseps}{<_{\epsilon \rightarrow 0^+ }}
\newcommand{\gteps}{>_{\epsilon \rightarrow 0^+ }}
\usepackage{nicefrac}

\newcommand{\CLS}{\mathsf{CLS}}
\newcommand{\FIXP}{\mathsf{FIXP}}
\newcommand{\Pol}{\mathsf{P}}

\crefname{question}{question}{questions}

\Crefname{question}{Question}{Questions}

\usepackage{authblk}

\title{The Complexity of Equilibrium Refinements in Potential Games\footnote{Authors ordered alphabetically. The previous version of this preprint contained results concerning normal-form proper equilibria; these results have now been extended and moved to a separate paper.}}

\author[1]{Ioannis Anagnostides}
\author[1]{Maria-Florina Balcan}
\author[1]{Kiriaki Fragkia}
\author[1,3]{Tuomas Sandholm}
\author[1]{Emanuel Tewolde}
\author[2]{Brian Hu Zhang}

\affil[1]{Carnegie Mellon University}
\affil[2]{Massachusetts Institute of Technology}
\affil[3]{Additional affiliations: Strategy Robot, Inc., Strategic Machine, Inc., Optimized Markets, Inc.}
\affil[ ]{}
\affil[ ]{\texttt{\{ianagnos,ninamf,kiriakif,sandholm,etewolde\}}\texttt{@cs.cmu.edu}, \texttt{zhangbh}\texttt{@csail.mit.edu}}

\begin{document}

\begin{titlepage}
\maketitle
\pagenumbering{gobble}

\begin{abstract}
Refinements of the Nash equilibrium---most notably Selten's trembling-hand perfect equilibrium and Myerson's proper equilibrium---are at the heart of microeconomic theory, addressing some of the deficiencies of Nash's original concept. While the complexity of computing equilibrium refinements has been at the forefront of algorithmic game theory research, it has remained open in the seminal class of potential games; we close this fundamental gap in this paper.

We first establish that computing a pure(-strategy) perfect or proper equilibrium is $\mathsf{PLS}$-complete in concise potential games in normal form. For pure perfect equilibria, we extend this result to general polytope games, which includes extensive-form games. We next turn to more structured classes of games, namely symmetric network congestion and symmetric matroid congestion games. For both classes, we show that a pure perfect equilibrium can be computed in polynomial time, strengthening the existing results for pure Nash equilibria. More broadly, we make a connection between strongly polynomial-time algorithms and efficient perturbed optimization using fractional interpolation. On the other hand, we establish that, for a certain class of potential games, there is an exponential separation in the length of the best-response path between perfect and Nash equilibria. Finally, for mixed strategies, we prove that computing a point geometrically near a perfect equilibrium requires a doubly exponentially small perturbation even in $3$-player potential games in normal form. As a byproduct, this significantly strengthens and simplifies a seminal result of Etessami and Yannakakis (FOCS '07). On the flip side, in the special case of polymatrix potential games, we show that equilibrium refinements are amenable to perturbed gradient descent dynamics, thereby belonging to the complexity class $\mathsf{CLS}$. This provides a principled and practical way of refining the landscape of gradient descent in constrained optimization.
\end{abstract}

\newpage
\tableofcontents

\end{titlepage}
\pagenumbering{arabic}
\section{Introduction}

A forceful critique of the Nash equilibrium~\citep{Nash50:Equilibrium}---the predominant solution concept in game theory---is centered on its non-uniqueness: a game may admit multiple Nash equilibria, with some being more sensible than others. A pathbreaking line of work in economics has put forward \emph{equilibrium refinements} as a compelling antidote to this conundrum. By now, an entire hierarchy of equilibrium concepts has emerged that differ depending on the game representation and the refinement criterion~\citep{VanDamme91:Stability}. Perhaps the most popular refinement of the Nash equilibrium is Selten's \emph{(trembling-hand) perfect equilibrium}~\citep{Selten75:Reexamination}. The impact of this concept is witnessed by the fact that Selten, together with Nash and Harsanyi, went on to win the Nobel prize in economics for this work.

From an algorithmic standpoint, a flurry of results has established that, for the most part, identifying such equilibrium refinements is no harder than computing Nash equilibria (highlighted at greater length in~\Cref{sec:related}). In particular, the rough complexity landscape that has come to light after decades of research is that equilibrium computation---under any of the usual refinements and game representations---in two-player zero-sum games is in $\Pol$, in two-player general-sum games is $\PPAD$-complete, whereas in $n$-player games is $\FIXP$-complete (\Cref{sec:related}).

However, surprisingly, the complexity of equilibrium refinements has remained unexplored in the seminal class of \emph{potential games}~\citep{monderer1996potential,Rosenthal73:Class}. Such games encompass a diverse panoply of strategic interactions including \emph{tacit coordination} among agents with aligned interests, a central problem in the social sciences with ubiquitous applications~\citep{Schelling60:Strategy}; routing problems in large communication networks in the absence of a centralized authority, such as the Internet~\citep{Roughgarden00:How}; and more broadly, resource allocation problems~\citep{Rosenthal73:Class}. Characterizing and indeed enhancing the quality of Nash equilibria in potential games---primarily through the lens of \emph{price of anarchy}~\citep{Koutsoupias99:Worst}---has been at the forefront of the algorithmic game theory agenda, as discussed further in~\Cref{sec:related}. We take a complementary perspective on this problem revolving around equilibrium refinements.

As it turns out, the complexity of computing a Nash equilibrium in potential games differs drastically from that in general games. First, unlike general games, potential games always admit a \emph{pure} Nash equilibrium, that is, a strategy profile in which no player randomizes. \citet{Fabrikant04:Complexity} famously showed that computing a pure Nash equilibrium is complete for the class \PLS, which stands for \emph{polynomial local search}~\citep{Johnson88:How}. Second, computing an approximate \emph{mixed} Nash equilibrium is complete for $\PLS \cap \PPAD = \CLS$~\citep{Fearnley23:Complexity,Babichenko21:Settling}.

\subsection{Overview of our results}

Our contribution here is to characterize the complexity of various equilibrium refinements in potential games. Our main results are summarized in~\Cref{tab:results}.

The first takeaway of our results is that, for the most part, equilibrium refinements in potential games share the same computational properties as Nash equilibria, both on the complexity and the algorithmic front. Indeed, we first show that computing \emph{pure(-strategy)} perfect or proper equilibria in concise potential games in normal form is in \PLS, thereby being polynomial-time equivalent to Nash equilibria. In fact, for perfect equilibria, our $\PLS$ membership applies beyond normal-form games, to general \emph{polytope games}; extensive-form games constitute a canonical member of polytope games. Second, we extend existing positive results for Nash equilibria to structured classes of games---namely matroid congestion games and symmetric network congestion games---to perfect equilibria. Third, for computing \emph{mixed-strategy} perfect and proper equilibria, we establish $\CLS$ membership in polymatrix games. We obtain this through a natural perturbed version of gradient descent that provides a principled and practical way of refining the solution space in constrained optimization problems, complementing a long line of work in contemporary optimization theory (\Cref{sec:related}). Fourth, we show that existing lower bounds on the price of anarchy of Nash equilibria can be extended to both perfect and proper equilibria.

On the flip side, some of our results turn this narrative around. For example, we establish exponential convergence separations between Nash and perfect (or proper) best-response dynamics. Surprisingly, this exponential separation goes both ways: there are games in which convergence to Nash equilibria is fast but to perfect (or proper) equilibria slow, and \emph{vice versa}. Also, on the price of anarchy front, we find that perfect or proper equilibria can be arbitrarily more efficient, and this is so even in comparison to the equilibria typically reached by learning dynamics (such as gradient descent, \Cref{fig:quiver_perfect}).

{
\begin{table}[!ht]
    \centering
    \footnotesize
    \caption{Overview of our main results.}
    \begin{tabular}{@{}l|cc@{}}
\toprule
\textbf{Game class / Problem} & \textbf{Perfect} & \textbf{Proper} \\
\midrule
\multicolumn{3}{l}{\textit{Pure strategies in potential games}} \\
\quad Concise normal form & $\PLS$-c (\Cref{theorem:PLS-complete}) & $\PLS$-c (\Cref{theorem:PLS-complete}) \\
\quad \multirow{1}{*}{Extensive form} & \multirow{1}{*}{$\PLS$-c (\Cref{th:efpe,th:qpe})} & -- \\
\quad Polytope games & $\PLS$-c (\Cref{theorem:polytope-perfect}) & -- \\
\midrule
\multicolumn{3}{l}{\textit{Mixed strategies in potential games}} \\
\quad Polymatrix & $\CLS$ (\Cref{cor:CLS}) & $\CLS$ (\Cref{cor:CLS}) \\
\quad \multirow{2}{*}{3-player identical-interest} & \multicolumn{2}{c}{$\Omega(1)$ distance from $\epsilon$-PE to exact PE even for $\epsilon = 1/2^{2^n}$} \\
 & \multicolumn{2}{c}{(\Cref{theorem:doublyexponentiallysmall,theorem:3playerdoublyexp})} \\
\midrule
\multicolumn{3}{l}{\textit{Tractable subclasses}} \\
\quad Symmetric matroid congestion & $\P$ (\Cref{theorem:polysteps}) & --- \\
\quad Symmetric network congestion & $\P$ (\Cref{theorem:symmetricnetwork}) & --- \\
\midrule
\multicolumn{3}{l}{\textit{Verification complexity}} \\
\quad 3-player identical-interest & $\NP$-hard (\Cref{prop:hard-verif}) & --- \\
\midrule
\multicolumn{3}{l}{\textit{Exponential convergence separations for best-response dynamics}} \\
\quad Poly. to Nash / Exp. to refinement & \Cref{theorem:expsep} & \Cref{cor:exppaths for proper} \\
\quad Poly. to refinement / Exp. to Nash & \Cref{theorem:expsep-opo} & \Cref{cor:exppaths for proper} \\
\midrule
\multicolumn{2}{l}{\textit{Price of anarchy}} \\
\quad Polynomial degree-$d$ congestion games & $d^{\Omega(d)}$ (\Cref{theorem:informal-poa-lower}) & $d^{\Omega(d)}$ (\Cref{theorem:informal-poa-lower}) \\
\bottomrule
\end{tabular}
    \label{tab:results}
\end{table}
}

The second main takeaway of our results is that certain seminal hardness results in general-sum games can be extended to potential games. First, even in a three-player potential game, we show that one needs to take \emph{doubly exponentially} small precision to guarantee geometric closeness to Nash equilibria. (We obtain this as a byproduct of a lower bound relating $\epsilon$-perfect to exact perfect in potential games; \Cref{theorem:doublyexponentiallysmall}.) This significantly refines a key result in the landmark paper of~\citet{Etessami07:Complexity}, which was proven in general-sum games and with a significantly more involved argument. In a similar vein, we establish \NP-hardness of verifying whether a purported perfect equilibrium is indeed one even in potential games. This strengthens the result of \citet{Hansen10:Computational}. Given the structure present in potential games \emph{vis-\`a-vis} general-sum games, it was hitherto highly unclear whether intractability persists for those central problems.

Taken together, our results shed light on the complexity of equilibrium refinements in potential games. Given the tremendous impact of equilibrium refinements in economics and the central role of potential games in algorithmic game theory, we hope our results will stimulate further exploration at this fertile intersection.

\section{Preliminaries}

Before we present our results in more detail, we provide the necessary notation and game theory background, with an emphasis on equilibrium refinements (\Cref{sec:eq-ref}).

\subsection{Game classes}

A \emph{normal-form} game $\Gamma$ consists of a finite set of players $[n] = \{1, \ldots n\}$. Each player has a finite set of actions $\cA_i$ that consists of pure strategies. We denote  $m_i = |\cA_i|$ and $m = \max_{i \in [n]} |\cA_i|$. The utility function of player $i$, denoted as $u_i:\prod_{j=1}^{n} \cA_j \to \R$ maps a strategy profile, $ (a_1, \ldots, a_n) \in \prod_{j=1}^{n} \cA_j := \cA$, to player $i$'s payoff $u_i(a_1, \ldots, a_n)$. We represent $\Gamma$ as the tuple $\Gamma := (\cA_1, \ldots, \cA_n, u_1, \ldots, u_n)$.

Players can mix their strategies by playing a probability distribution $\vx_i = (\vx_{i}(a_i))_{a_i \in \cA_i}$ over their action set. We will use the notation $\Delta(S)$ to denote the set of probability distributions on set $S$, and abbreviate $\Delta([n]) = \Delta(n)$. We denote with $\cX_i = \Delta(\cA_i)$ the set of such mixed strategies for player $i$, which we will call the mixed strategy space for player $i$. By $\cX = \prod_{i=1}^n \cX_i$ we denote the mixed strategy space for a game. For a mixed strategy profile $\vx = (\vx_1, \ldots, \vx_{n}) \in \prod_{i=1}^{n} \Delta(\cA_i)$, we denote by $\vx_{-i}$ the strategy profile of every player except player $i$ and (with an overload in notation), we define the expected utility for player $i$ as
\begin{align*}
    u_i(\vx_1, \ldots, \vx_n) := \sum_{(a_1 \ldots, a_n) \in \cA} u_i(a_1, \ldots, a_n) \prod_{j = 1}^n \vx_j(a_j).
\end{align*}

\paragraph{Potential games} Informally, a  potential game is identified by a potential function which maps strategy profiles to real values and which increases by the same amount as the utility of any unilaterally deviating player. Identical interest games provide a canonical example of potential games, in which the common utility plays the role of the potential.

\begin{definition}[Potential game]
    \label{def:pot}
    A potential game is a game for which there exists a potential function, $\Phi: \cA \to \R$, such that for  any player $i \in [n]$, any strategy $(a_1, \dots, a_n) \in \cA$, and any unilateral deviation $a'_i \in \cA_i$,
    \begin{align*}
        \Phi(a'_i, a_{-i}) - \Phi(a_i, a_{-i}) =  u_i(a'_i, a_{-i}) - u_i(a_i, a_{-i}).
    \end{align*}
\end{definition}

\paragraph{Polymatrix games} Another class of games that we will refer to in some of our results is \emph{polymatrix games}~\citep{Cai16:Zero}.

\begin{definition}[Polymatrix game]
    \label{def:polymatrix}
    We let $G=(V, E)$ denote a simple, undirected graph where $V = [n]$ is the set of vertices and $E$ is the set of edges. Each player $i \in [n]$ has a pure strategy set $\cA_i$ and each edge $\{i, j\} \in E$ defines a two-player normal form game with payoff matrices $\mat{P}_{i, j} \in \Q^{|\cA_i| \times |\cA_j|}$ and $\mat{P}_{j, i} \in \Q^{|\cA_j| \times |\cA_i|}$ for player $i$ and $j$ respectively. A polymatrix game $\Gamma = (G, (\mat{P}_{i, j}, \mat{P}_{j, i})_{\{i, j\} \in E})$ is an $n$-player game with the following utility function, $u_i: \cX \to \R$, for each player $i \in V$
    \begin{align*}
        u_i(\vx_1, \ldots, \vx_n) = \sum_{j \in N_i} \vx_i^\top \mat{P}_{i, j} \vx_j,
    \end{align*}
    where $\vx_i \in \cX_i$ and $N_i$ denotes the set of neighbors of $i \in V$.
\end{definition}

 \paragraph{Congestion Games} A congestion game is a tuple $\Gamma = \left(\cR, (\cA_i)_{i \in [n]}, (d_r)_{r \in \cR} \right)$, where $[n]$ is the set of players, $\cR$ is a set of resources, and the strategy space for a player $i \in [n]$ is a collection of subsets of resources, $\cA_i \subseteq 2^{\cR}$. 
    Each resource $r \in \cR$ has an associated delay function
    $d_r: \N\to \N$. If all players have the same strategy space, we call the congestion game symmetric.

    We write a strategy profile as $S = (S_1, \dots, S_n)$, where each player $i \in [n]$ is playing strategy $S_i \in \cA_i$. Given such a strategy profile $S$, we define $n_r(S) = |\{i \in [n]: r \in S_i\}|$ to be the number of players using resource $r$ in $S$. We assume that each player $i$ is rational and is trying to minimize her own cost, which for a strategy profile $S$, is described by the cost function $c_i(S) = \sum_{r \in S_i} n_r(S)$. 
    
    Given a strategy profile, $S$, a strategy $S^*_i \in \cA_i$ is a best-response of player $i$ to $S$ if $c_i(S^*_i, S_{-i}) \leq c_i(S'_i, S_{-i})$ for all $S'_i \in \cA_i$. A strategy profile, $S^*$, is a pure Nash equilibrium if no player can decrease her cost by changing her strategy. That is, for all players $i \in [n]$ and all strategies $S_i \in \cA_i$, $c_i(S_i, S^*_{-i}) \geq c_i(S^*)$.

    By \citet{Rosenthal73:Class}, we know that congestion games are potential games with potential function \begin{align*}
        \Phi(S) = \sum_{r \in \cR}\sum_{i=1}^{n_r(S)}d_r(i).
    \end{align*}

\paragraph{Extensive-form games} Extensive-form games are common representations of sequential interactions featuring imperfect information~\citep{Shoham08:Multiagent}, represented compactly through a game tree. We present a formal introduction in \Cref{sec:efg}, as it is not essential for the purposes of the main body.

\paragraph{Polytope games} In polytope games each player selects a strategy from a (convex) polytope described explicitly through a set of linear inequalities; the number of vertices is generally exponential---as is the case for the \emph{sequence-form polytope} in extensive-form games---so our approach for the normal-form case needs to be refined. Canonical examples are (network) congestion games~\citep{Roughgarden00:How} and, indeed, extensive-form games. We defer a formal introduction of polytope games to \Cref{sec:polytope}.

\subsection{Equilibrium refinements}
\label{sec:eq-ref}

The standard solution concept in game theory---and the main reference point of equilibrium refinements---is the (mixed) Nash equilibrium~\citep{Nash50:Equilibrium}, recalled below.

\begin{definition}[Nash equilibrium]\label{def:nash}
    A mixed strategy profile $\vx^* = (\vx^*_1, \ldots, \vx^*_n) \in \cX$ is called a (mixed) Nash equilibrium if  it holds that $u_i(\vx^*_i, \vx^*_{-i}) \geq u_i(\vx_i, \vx^*_{-i})$ for all players $i$ and all possible deviations $\vx_{i} \in \cX_i$.
    Equivalently, we can write $u_i(\vx^*_i, \vx^*_{-i}) \geq u_i(a_i, \vx^*_{-i})$ for all players $i$ and all possible pure strategy deviations $a_{i} \in \cA_i$.
\end{definition}

\paragraph{Perfect equilibrium} 

Moving beyond Nash equilibria, we now give a brief overview of the basic equilibrium refinements that we consider, together with some motivating examples specifically in potential games. We begin with the usual normal-form representation of games, and then extend our scope to other more involved representations.

The (trembling-hand) perfect equilibrium was put forward by~\citet{Selten75:Reexamination} so as to exclude certain unreasonable Nash equilibria, such as ones supported on strictly dominated actions away from the equilibrium path. At a high level, perfect equilibria account for players' mistakes---hence the term ``trembling-hand.'' In particular, a strategy profile qualifies as a perfect equilibrium if it arises as the limit point of a sequence of Nash equilibria in a perturbed game wherein each action is played with some positive probability (\Cref{def:PE} formalizes this). 

\begin{figure}[!h]
\centering
\begin{tabular}{c c c}
& \textsf{C1} & \textsf{C2} \\
\textsf{R1} & 1, 1 & 0, 0 \\
\textsf{R2} & 0, 0 & 0, 0 \\
\end{tabular}
\caption{A $2 \times 2$ identical-interest game in normal form. It has exactly two Nash equilibria: (\textsf{R1}, \textsf{C1}) and (\textsf{R2}, \textsf{C2}). Of the two, only (\textsf{R1}, \textsf{C1}) is a perfect equilibrium.}
\label{fig:nash-vs-perfect}
\end{figure}

\Cref{fig:nash-vs-perfect} portrays a simple illustrative example. Of the two equilibria of the game, (\textsf{R1}, \textsf{C1}) is the only sensible one---\textsf{R2} is weakly dominated by \textsf{R1} and \textsf{C2} is weakly dominated by \textsf{C1}. As soon as one accepts that all actions are played with some positive probability (for example, through a mistake or a ``tremble''), (\textsf{R1}, \textsf{C1}) emerges unequivocally as the unique perfect equilibrium---in the limit as the trembles vanish, that is. This is a rather prosaic and brittle example, but we shall analyze more interesting examples in the sequel.

Formally, an \emph{$\epsilon$-perfect equilibrium} requires that every player plays a fully mixed strategy, but only pure strategies that are best responses get played with probability greater than $\epsilon$.


\begin{definition}[$\epsilon$-perfect equilibrium]\label{def:eps-perfect}
    For some $\epsilon >0$ a mixed strategy profile $\vx$ is an $\epsilon$-perfect equilibrium if it is fully mixed and
    \begin{align*}
        u_i(a_i, \vx_{-i}) < u_i(a_i', \vx_{-i}) \implies \vx_i(a_i) \leq \epsilon \; \; \text{for all players} \; i \in [n] \; \text{and actions} \; a_i, a_i' \in \cA_i.
    \end{align*}
\end{definition}

Let $\cX_{i}^{(\epsilon)} \defeq \{\vx_i \in \cX_i \; \text{such that} \; \vx_i(a_i) \geq \epsilon \; \forall a_i \in \cA_i\} $ and $\cX^{(\epsilon)} = \prod_{i=1}^{n} \cX_i^{(\epsilon)}$. We will call a strategy $\vx_i \in \cX^{(\eps)}_i$ {\em $\eps$-pure} if it is as pure as possible given the perturbed game, that is, if $\vx_i(a_i) = \eps$ for all but one $a_i \in \cA_i$. We will call an $\eps$-pure profile $\vx$ an {\em $\eps$-pure equilibrium} if it is $\eps$-perfect. Asserting that players must make mistakes with small probability gives rise to the following notion.

\begin{definition}[Perfect equilibrium]
    \label{def:PE}
    Consider a game $\Gamma$. A {\em trembling-hand perfect equilibrium}, or simply a {\em perfect equilibrium}, of $\Gamma$ is a limit point of a sequence $\{\vx^{(\epsilon)} \in \cX^{(\epsilon)}\}_{\epsilon \to 0^{+}}$ where $\vx^{(\epsilon)}$ is an $\epsilon$-perfect equilibrium of $\Gamma$.
\end{definition}
Due to \citet{Selten75:Reexamination}, we know that every normal form game has at least one perfect equilibrium and every perfect equilibrium must be a Nash equilibrium.

\begin{definition}[$\epsilon$-perturbed game for perfect equilibria]\label{def:perfect-perturbed}
    Consider a finite game $\Gamma$. For some $\epsilon \in (0, 1)$, we define the perturbed game to be $\Gamma^{(\epsilon)} =(\cX_1^{(\epsilon)}, \ldots, \cX_n^{(\epsilon)}, u_1, \ldots, u_n)$.
\end{definition}

From \Cref{def:nash,def:eps-perfect}, it follows that every Nash equilibrium of the perturbed game $\Gamma^{(\epsilon)}$ is an $\eps$-perfect equilibrium of the original game $\Gamma$.

\paragraph{Proper equilibrium}  

Moving on, \emph{proper equilibria}, introduced by another Nobel laureate, namely Roger \citet{Myerson78:Refinements}, were in turn shown to refine perfect equilibria. The high-level idea of the definition is this: although players can make mistakes, as before, they now do so in a somewhat rational way. In particular, some mistakes are more costly than others, and so should be made with a smaller probability. As observed by~\citet{Kohlberg86:strategic}, a proper equilibrium can be obtained as a limit point of Nash equilibria in a sequence of $\epsilon$-perturbed games in which players select strategies from a certain \emph{permutahedron}---namely, the convex hull of permutations of $(1, \epsilon, \epsilon^2, \dots, \epsilon^{m-1})$, up to some normalization factor; this means that the action resulting in the highest utility is to be allotted probability roughly $1$, the second highest roughly $\epsilon$, and so on, so that progressively worse actions end up being played with gradually smaller probability (\Cref{def:proper}).

\begin{figure}[!ht]
    \centering
    \begin{tabular}{c c c c}
    & \textsf{C1} & \textsf{C2} & \textsf{C3} \\
    \textsf{R1} & 1, 1 & 0, 0 & -9, -9 \\
    \textsf{R2} & 0, 0 & 0, 0 & -7, -7 \\
    \textsf{R3} & -9, -9 & -7, -7 & -7, -7\\
    \end{tabular}
    \caption{A $3 \times 3$ identical-interest game in normal form devised by~\citet{Myerson78:Refinements}. }
    \label{fig:perfect-vs-proper}
\end{figure}

A motivating example, due to~\citet{Myerson78:Refinements}, is given in~\Cref{fig:perfect-vs-proper}. As in~\Cref{fig:nash-vs-perfect}, it seems evident that the unique outcome of the game should be (\textsf{R1}, \textsf{C1}); after all, only strictly dominated actions were inserted. And yet, this is not in accordance with the set of perfect equilibria of that game: (\textsf{R2}, \textsf{C2}) is now in fact a perfect equilibrium. (To see this, one can take $\vx^{(\epsilon)}_1 = (\epsilon, 1 - 2 \epsilon, \epsilon)$ and $\vx^{(\epsilon)}_2 = (\epsilon, 1 - 2 \epsilon, \epsilon)$. It then follows that $(\vx_1^{(\epsilon)}, \vx_2^{(\epsilon)} )$ is an equilibrium in the $\epsilon$-perturbed game \emph{\`a la} Selten.) On the other hand, \citet{Myerson78:Refinements} observed that the only proper equilibrium is (\textsf{R1}, \textsf{C1}). Peter Bro Miltersen has anecdotally referred to proper equilibria as ``the mother of all refinements'' in normal-form games~\citep{Etessami21:Complexity}; it will indeed be the most refined concept that we examine.

Formally, we start with the following definition.

\begin{definition}[$\epsilon$-proper equilibrium]
    For some $\epsilon >0$ a mixed strategy profile $\vx$ is an $\epsilon$-proper equilibrium if it is fully mixed and
    \begin{align*}
        u_i(a_i, \vx_{-i}) < u_i(a_i', \vx_{-i}) \implies \vx_i(a_i) \leq \epsilon \vx_i(a_i') \; \; \text{for all players} \; i \in [n] \; \text{and actions} \; a_i, a_i' \in \cA_i.
    \end{align*}
\end{definition}

The definition of an $\epsilon$-proper equilibrium naturally leads to the following notion of a perturbed game due to~\citet{Kohlberg86:strategic}, which differs from that in~\Cref{def:perfect-perturbed}. Here we will overload notation, using $\cX_i^{(\epsilon)}$ to denote the perturbed strategy set for both perfect and proper equilibria, as the meaning will be clear from the context.

\begin{definition}[$\epsilon$-perturbed game for proper equilibria]\label{def:proper-perturbed} \sloppy Consider a finite game $\Gamma$. For some $\epsilon \in (0, 1)$ let the strategy space $\cX_i^{(\epsilon)}$ be
\begin{align*}
    \cX_i^{(\epsilon)} = \text{conv}\left(\biggl\{ \frac{1-\epsilon}{1-\epsilon^{m_i}} ( \epsilon^{\pi(0)}, \epsilon^{\pi(1)}, \ldots, \epsilon^{\pi(m_i-1)}): \pi \in S_{m_i-1}  \biggr\}\right), 
\end{align*}
where conv denotes the convex hull of a set of vectors, $S_{m_i-1}$ is the set of all permutations of $\{0, \ldots, m_i-1\}$ and $\pi \in S_{m_i-1}$ is a permutation i.e., a bijection from $\{0, \ldots, m_i-1\}$ to itself. We then define the perturbed game as $\Gamma^{(\epsilon)} = (\cX_1^{(\epsilon)}, \ldots, \cX_n^{(\epsilon)}, u_1, \ldots, u_n)$.
\end{definition}
In the above definition, there is no need to redefine the utilities, as mixed strategies over mixed strategies of the original game are still just mixed strategies.

\begin{definition}[Proper equilibrium]
    \label{def:proper}
    A mixed strategy profile $\vx$ is a proper equilibrium if it is the limit point of a sequence $\{\vx^{(\epsilon)} \in \cX^{(\eps)}\}_{\eps \to 0^{+}}$ where $\vx^{(\epsilon)}$ is an $\eps$-proper equilibrium of $\Gamma$. 
\end{definition}

Due to \citet{Myerson78:Refinements}, we know that in any normal-form game there exists at least one proper equilibrium. Furthermore, the proper equilibria form a subset of the perfect equilibria, which in turn are a subset of the Nash equilibria. \citet{Kohlberg86:strategic} provide a constructive proof of the existence of a proper equilibrium in a normal-form game. This proof relies on computing a Nash equilibrium of the perturbed game, $\Gamma^{(\epsilon)}$, which in turn corresponds to an $\epsilon$-proper equilibrium of the original game.

\paragraph{Equilibrium refinements in extensive-form games} In extensive-form games, one can employ the previous notion of perfection to the induced normal-form game, giving rise to what is known as \emph{normal-form perfect equilibrium (NFPE)}. But NFPE is not the most attractive refinement in extensive-form games~\citep{Etessami21:Complexity}. We instead examine the more refined notions of \emph{extensive-form perfect equilibria (EFPEs)} and \emph{quasi-perfect equilibria (QPEs)}; these are incomparable with each other, in that an EFPE need not be a QPE and a QPE need not be an EFPE~\citep{Mertens95:Two}. We refrain from formally introducing those concepts at this point. Suffice it to say that, for QPEs, we make use of a characterization due to~\citet{Gatti20:Characterization} in multi-player extensive-form games, showing that they can be obtained as limit points of a sequence of Nash equilibria of a certain class of perturbed games in \emph{sequence form} (\emph{cf.}~\citet{Miltersen10:Computing}). Similarly, for EFPEs, there is again a characterization as limit points of Nash equilibria of a certain class of perturbed games in sequence form~\citep{Farina17:Extensive,Blume91:Lexicographic}; this was used by~\citet{Farina17:Extensive} to place EFPEs in \PPAD. The precise definitions of the perturbed games are deferred to a later section.

To make a connection with the foregoing definitions in normal-form games, we point out that one can compute a quasi-perfect equilibrium of an extensive-form game~\citep{Miltersen10:Computing} by identifying a proper equilibrium of the induced normal-form game~\citep{vanDamme84:relation}, since the former are a superset of the latter.

\section{Multilinear circuits, search problems, and symbolic computation through polynomial interpolation}
\label{sec:model}



In this section, we provide some essential groundwork for our results. We first introduce the notion of a multilinear arithmetic circuit, leading to the notion of a concise potential game (\Cref{def:concise_potential}). We then address a subtle issue in the definition of search problems, which is especially relevant for studying the complexity of equilibrium refinements. Finally, \Cref{sec:symbolc-comp} establishes two key results that we rely upon: i) efficient evaluation of multilinear arithmetic circuits with input a symbolic polynomial via polynomial interpolation, and ii) symbolic comparison of polynomials.

\paragraph{Multilinear arithmetic circuits} Numbers in inputs will always be rational numbers represented as fractions in reduced form, given in binary. In general, we will use $\size(a)$ to denote the representation size of the object $a$, \ie, the number of bits needed to represent it as input. Our inputs, \ie, utility or potential functions, will be multilinear functions $f : \cX \to \R$, usually provided as {\em arithmetic circuits}. To enforce multilinearity, we will only allow addition gates, multiplication gates, and rational constants; and we will insist that different inputs to any given multiplication gate must involve disjoint sets of players and rational constants. (In particular, the sub-DAGs induced by the inputs from a multiplication gate must be disjoint.) For the remainder of the paper, we will call circuits that satisfy the above assumptions {\em multilinear arithmetic circuits}. The output of such a circuit can be evaluated efficiently, as formalized below; the proof is in~\Cref{sec:furtherproofs}.

\begin{restatable}{lemma}{efficmult}
    \label{lemma:efficient-eval}
    Multilinear arithmetic circuits can be evaluated efficiently. That is, given a multilinear arithmetic circuit $f : \R^M \to \R$ and input $\vx \in \Q^M$, the output $f(\vx) \in \Q$ can be computed in time $\poly(\size(f), \size(\vx))$.
\end{restatable}

\begin{remark}
    Multilinear arithmetic circuits immediately capture another natural representation: $f$ is given as an explicit multivariate polynomial.
\end{remark}

\begin{definition}[Concise game]\label{def:concise_game}
    A concise game $\Gamma = (\cA_1, \ldots, \cA_n, u_1, \ldots, u_n)$ is a normal-form game in which each utility function, $u_i: \cX \to \R$, is provided as a multilinear arithmetic circuit.
\end{definition}

\begin{definition}[Concise potential game]\label{def:concise_potential}
    A concise potential game $\Gamma = (\cA_1, \ldots, \cA_n, \Phi)$ is a normal-form potential game in which the potential function, $\Phi: \cX \to \R$, is provided as a multilinear arithmetic circuit.
\end{definition}



The conciseness assumption encompasses many classes of succinct, multi-player games~\citep{Papadimitriou08:Computing}, including polymatrix games, bounded-degree graphical games, and anonymous games with a constant number of actions. We relax some of these assumptions when we study general polytope games.

\subsection{Search problems}
\label{sec:search}


Let $\Sigma = \{0, 1\}$ and $\Sigma^*$ denote the set of all finite binary strings. A {\em search problem} is defined by a relation $R \subseteq \Sigma^* \times \Sigma^*$. An algorithm {\em solves} the search problem $R$ if, when given input $x \in \Sigma^*$, the algorithm outputs some $y$ such that $(x, y) \in R$, or (correctly) asserts that no such $y$ exists.

In the literature on complexity of search problems, it is standard to assume that the relation $R$ is efficiently computable, that is, there is a polynomial-time algorithm that, on input $(x, y)$, checks whether $(x, y) \in R$. The set of search problems with this property is often called \FNP. In this paper, however, we will commonly deal with search problems for which a polynomial-time algorithm exists, yet solutions cannot be efficiently verified. As a simple example of such a problem, consider $$R_\textsf{HALT} = \{ (x, y) : y \text{ is a Turing machine that halts on the empty tape} \}.$$ Then there is a constant-time algorithm for $R_\textsf{HALT}$: given any input $x$, ignore $x$ and output a trivial TM that halts immediately. However,  verifying a solution $(x, y)$ is equivalent to the halting problem, which is undecidable. That is, under the usual definition of $\FNP$, $R_\textsf{HALT} \notin \FNP$. Therefore, in this paper we will define the relevant complexity classes in a way that is slightly different from the usual definitions.\footnote{Our definition of \FNP\ is also used by \citet{Johnson88:How}, who call this class $\NP_S$, where the $S$ stands for {\em search}, but to our knowledge is rarely or never used by authors since then. Since reductions are defined the same way regardless, to our knowledge, all known results surrounding the complexity of search problems hold for these modified definitions as well.}
\begin{definition}[\citealp{Johnson88:How}]
    \label{def:TMnondeter}
    A search problem $R$ is in \FNP\ if there exists a nondeterministic polynomial-time Turing machine $M$ such that
    \begin{enumerate}
        \item the language recognized by $M$ is $D_R := \{ x : \exists y \text{ s.t. } (x, y) \in R\}$, and 
        \item on any accepting path of $M$, the TM outputs (to its tape) a $y$ satisfying $(x, y) \in R$.
    \end{enumerate}
\end{definition}
This definition does {\em not} imply the usual verifier definition of \FNP: for example, $R_\textsf{HALT} \in \FNP$ by this definition, despite not being efficiently verifiable. However, we argue that it is in a sense more natural: is $R_\textsf{HALT}$ really a hard problem (not even in \FNP\ under the more common definition), when it has a trivial constant-time algorithm?
In our paper, we will later see that many problems surrounding potential games are of this nature, for example, it is easy to find a perfect equilibrium of a potential game with a constant number of players (\Cref{prop:constantplayers}), but hard to verify such an equilibrium (\Cref{prop:hard-verif}).

Other classes of search problems can be defined by specifying complete problems for them. 
One way to define such problems is using Boolean circuits. We define a Boolean circuit $C: \{0, 1\}^k \to \{0, 1\}^k$ as a function that can use the logic gates (AND), (OR),  and (NOT), denoted as $\land$, $\lor$, and $\lnot$ respectively.

\paragraph{Complexity classes} We continue by defining the complexity class $\PLS$ \cite{Johnson88:How} that contains \emph{local search problems}. Informally, $\PLS$ contains all local search problems with neighborhoods that are searchable in polynomial time. 
\begin{definition}\label{def:pls}
    The complexity class $\PLS$ is the class of all total search problems that are reducible to {\sc LocalOpt}: given 
    two Boolean circuits $ S, V: [2^k] \to [2^k]$, find $a \in \cA$ such that $V(S(a)) \le V(a)$.
\end{definition}

The complexity class $\CLS$ \cite{Daskalakis11:Continuous, Fearnley23:Complexity} contains all \emph{continuous local search problems}. Informally, it contains all search problems seeking an approximate local optimum of a continuous function.

\begin{definition}
    The class $\CLS$ is the class of all total search problems that are reducible to {\sc Continuous-LocalOpt}: given a precision parameter $\epsilon >0$, well-behaved\footnote{Informally, a circuit is well-behaved if it does not allow repeated squaring. For more details, we refer the reader to~\citet{Fearnley23:Complexity} and~\url{https://people.csail.mit.edu/costis/CLS-corrigendum.pdf}.} circuits $p: [0, 1]^k \to [0, 1]$ and  $g: [0, 1]^k \to [0, 1]$, and a Lipschitz constant $L>0$, find $\vx \in [0, 1]^k$ such that $p(g(\vx)) \geq p(\vx) -\epsilon$. Alternatively, find $\vx, \vy \in [0, 1]^k$ such that $|p(\vx)-p(\vy)| \geq L\|\vx-\vy\|$ or $\|g(\vx)-g(\vy)\| \geq L\|\vx-\vy\|$.
\end{definition}

\subsection{Evaluating circuits on polynomials and comparison of polynomials}
\label{sec:symbolc-comp}

A key primitive that underpins our $\PLS$ membership results concerns the evaluation of circuits on inputs $\vx$ that are themselves polynomials of a single variable $\eps$. We first argue that such evaluations are well-behaved, in the sense that they will not result in intermediate computations with high representation size. Our proof makes use of polynomial interpolation.

\begin{lemma}\label{lem:polynomial circuit}
    \label{lemma:interpolation}
    Let $f : \R^M \to \R$ be a polynomial of degree $d_f$ represented by an arithmetic circuit, such that evaluating $f(\vx)$ for a rational vector $\vx \in \Q^M$ takes time $\poly(\size(f), \size(\vx))$. Let $\vx \in \Q[\eps]^M$ be a vector of length $M$ each of whose entries is a polynomial in $\eps$ of degree at most $d_x$. Then the output $f(\vx) \in \Q[\eps]$ can also be computed in time $\poly(d_f, d_x, \size(f), \size(\vx))$ 
\end{lemma}
\begin{proof}
    The composition $g = f(\vx) : \eps \mapsto f(\vx(\eps))$ is a polynomial $g : \R \to \R$, of degree at most $d_x \cdot d_f$. It is thus uniquely determined by $1+ d_x \cdot d_f$ values. Thus, $g$ can be evaluated as follows: compute $g(\eps) = f(\vx(\eps))$ for $\eps = 0, 1, 2, \dots, d_x \cdot d_f$, and then recover $g$ via polynomial interpolation, namely, 
    \begin{align*}
        g(\eps) = \sum_{j=0}^{d_x \cdot d_f} g(j) \prod_{i \ne j} \frac{\epsilon - i}{j - i}. \tag*\qedhere
    \end{align*}
\end{proof}

Another central primitive in our results concerns comparing polynomials. Specifically, let $p, q$ be univariate polynomials. We use the notation $p(\epsilon) \lesseps q(\epsilon)$ if there exists a sufficiently small $\epsilon_0 > 0$ such that $p(\epsilon) < q(\epsilon)$ for all $\epsilon \in (0, \epsilon_0)$. We observe that there is a simple, efficient algorithm for solving this problem.

\begin{lemma}\label{lem:integer potential}
    For any $L > 0$, there is a function $\psi : \Q[\eps] \to \N$ such that, for any two polynomials $p, q \in \Q[\eps]$ of representation size at most $L$, (1)  $\psi(p), \psi(q)$ are computable in $\poly(L)$ time, and (2) $p \lesseps q$ if and only if $\psi(p) < \psi(q)$.
\end{lemma}
\begin{proof}
   Comparing two polynomials $p, q$ with $\lesseps$ is equivalent to comparing their coefficients in lexicographic order. Since $\size(p) \le L$, the degree of $p$ is also at most $L$. First, consider comparing rational numbers $\alpha, \beta$ with size at most $L$. Such a rational number can have at most a $L$-bit denominator; therefore, $\alpha - \beta$ has denominator at most $4^L$. Thus, if $\alpha - \beta \ne 0$, then $\round(4^L \alpha) - \round(4^L \beta) \ne 0$. Moreover, we have $4^L |\alpha| < 4^L \cdot 2^L = 8^{L}$, so $0 \le 4^L(\alpha + 2^L) < 2 \cdot 8^L \le 16^L$. Thus, it suffices to take the map $\psi$ defined by
   \begin{align*}
       \psi\qty(\sum_{i=0}^L \alpha_i \eps^i) = \sum_{i=0}^L \round(4^L (\alpha_i + 2^L)) \cdot 16^{L(L-i)}. \tag*\qedhere
   \end{align*}
\end{proof}

\section{Main results}

This section contains a detailed overview of our results concerning the complexity of different equilibrium refinements in potential games under different game representations. We focus on highlighting the main results and providing detailed sketches of the key ideas in the proofs.

We first focus on \emph{concise} potential games. As discussed earlier, these are represented through \emph{multilinear arithmetic circuits}; multilinearity is enforced structurally by insisting that different inputs to any multiplication gate must involve disjoint sets of players. Multilinear arithmetic circuits can be evaluated efficiently on rational inputs (\Cref{lemma:efficient-eval}); thus we circumvent the issue of repeated squaring present when using general, unfettered arithmetic circuits~\citep{Fearnley23:Complexity}. We rely on conciseness because the explicit representation of a game---given as the entire payoff tensor---grows exponentially in the number of players. Also, we note that for explicitly represented identical-interest games, there is a straightforward algorithm for computing a Pareto-optimal perfect equilibrium (and other equilibrium refinements): identify a strategy profile corresponding to a maximum entry of the payoff tensor \emph{in the perturbed game} (\Cref{prop:constantplayers})---this algorithm quickly becomes inefficient as the number of players grows.

\subsection{Hardness of verification}

Despite this simple fact, and somewhat paradoxically, we show that even in $3$-player identical-interest games \emph{verifying} whether a pure strategy profile is a perfect equilibrium is \NP-hard.

\begin{theorem}
    \label{prop:hard-verif}
    Verifying whether a pure strategy is a perfect equilibrium in a $3$-player identical-interest game is \NP-hard.
\end{theorem}

This strengthens the hardness result of~\citet{Hansen10:Computational} (subsequently refined by~\citealp{Hansen19:Real}) pertaining to $3$-player general-sum games. Similar to their proof, the starting point of our reduction is the problem of computing the team minimax equilibrium (TME) value in adversarial team games (defined in~\Cref{sec:lowerboundverif}). We show that a gap-amplified hard TME instance allows a reduction to go through even when the players in the resulting game have identical interests. We present a formal proof of this result in \Cref{sec:lowerboundverif}.

Prompted by~\Cref{prop:hard-verif}, we now return to a subtle issue in the definition of $\FNP$ concerning the ability to verify purported solutions in polynomial time. As discussed in~\Cref{sec:search}, following the original treatment of~\citet{Johnson88:How}, we take $\FNP$ to contain all relations that can be accepted by a nondeterministic polynomial-time Turing machine (formalized in~\Cref{def:TMnondeter}). This definition does \emph{not} imply the usual efficient verification property (\Cref{sec:search}). In particular, \Cref{prop:hard-verif} does not preclude placing perfect equilibria---and refinements thereof---in $\FNP$.

\subsection{Complexity results for pure equilibria} 

Indeed, we begin by showing the following characterizations.

\begin{theorem}
    \label{theorem:PLS-complete}
    Computing a pure perfect equilibrium in concise potential games is $\PLS$-complete. The same holds for proper equilibria.
\end{theorem}

$\PLS$-hardness is readily inherited from existing, well-known results concerning pure Nash equilibria~\citep{Fabrikant04:Complexity}, so we discuss the proof of $\PLS$ membership. It is established by executing \emph{symbolic} best-response dynamics. Specifically, for a symbolic parameter $\epsilon > 0$, we run best-response dynamics on a perturbed game parameterized by $\epsilon$, which differs depending on the underlying equilibrium refinement notion. First, we have shown that by having access to a multilinear arithmetic circuit, one can still evaluate utilities symbolically in terms of $\epsilon$ through the use of polynomial interpolation (\Cref{lemma:interpolation}). Furthermore, when $\epsilon$ is taken arbitrarily small, it is possible to compare polynomials by contrasting their coefficients in lexicographic order (\Cref{lem:integer potential}). We are thus able to efficiently implement each iteration of the symbolic best-response dynamics when $\epsilon$ is small enough. Finally, the \emph{potential} that arises from these dynamics can itself be thought of as a polynomial in $\epsilon$, so that the key local improvement property attached to $\PLS$ can be established lexicographically with respect to the coefficients of the potential. We present a formal proof of \Cref{theorem:PLS-complete} in \Cref{sec:nfpe} and \Cref{sec:computingproper} for perfect and proper equilibria respectively.

As a non-trivial corollary of~\Cref{theorem:PLS-complete}, it follows that pure perfect and pure proper equilibria always exist in potential games; an earlier work in the economics literature by~\citet{Carbonell14:Refinements} had already shown the existence of a pure perfect equilibrium in potential games.

Taking a step further, we establish $\PLS$ membership beyond normal-form games. We first show that computing a pure EFPE or pure QPE in extensive-form games is $\PLS$-complete (\Cref{th:efpe,th:qpe}). The high-level argument is similar to that of~\Cref{theorem:PLS-complete}, but with an added complication concerning symbolic best responses in the perturbed game---be it the one arising from EFPEs or QPEs. This step is straightforward in normal-form games: for perfect equilibria it boils down to computing a maximum entry of the utility vector, while for proper equilibria---where the perturbed strategy set is a permutahedron---it reduces to sorting. In extensive-form games, symbolic best responses for both EFPEs and QPEs can be computed in polynomial time through a bottom-up traversal of the tree.

\paragraph{Polytope games} Turning to the more general setting of polytope games, we first make a conceptual contribution by introducing a notion of perfect equilibria in that setting. It is based on perturbing toward a point in the strict (relative) interior of the set (\Cref{def:pert-polytope}). We show that inclusion in $\PLS$ for perfect equilibria holds more broadly by merely assuming access to a linear optimization oracle for each strategy set---that is, a best response oracle.

\begin{restatable}{theorem}{polytopepls}
    \label{theorem:polytope-perfect}
    Finding a perfect equilibrium of a concise polytope potential game is in \PLS.
\end{restatable}

\subsection{Polynomial-time convergence in structured classes of games} Notwithstanding the $\PLS$-hardness of computing a pure-strategy Nash equilibrium, the question of carving out classes of games wherein best-response dynamics converge in polynomial time has received ample interest in the literature (\emph{e.g.}, we point to~\citet{Ackermann08:Impact,Caragiannis11:Efficient,Caragiannis17:Short}, and references therein). This begs the question: do these positive results carry over to perfect equilibria?

To begin with, we provide an affirmative answer in \emph{symmetric matroid congestion} games. Matroid congestion games were introduced by~\citet{Ackermann08:Impact} and have been studied extensively since then \cite{Balcan13:Price,Harks16:Logarithmic,Jong16:Efficiency,Hao24:Price}. The key premise is that each player's strategy set comprises the bases of an underlying matroid (\Cref{def:matroid,def:matroidcongestion}); the special case where the rank of the matroid is $1$---meaning that the strategy set of each player comprises single resources---is known as \emph{singleton games}~\citep{Ieong05:Fast}. On top of that, we posit symmetric strategy sets among the players; the reason we need this symmetry assumption will become clear shortly. Under these preconditions, we give a polynomial-time algorithm for computing perfect equilibria.

\begin{theorem}
    \label{theorem:polysteps}
    In symmetric matroid congestion games, symbolic best-response dynamics converge to a perfect equilibrium after a polynomial number of steps.
\end{theorem}

In proof, it suffices to show that the perturbed game is itself a matroid congestion game---since one can then appeal to the result of~\citet{Ackermann08:Impact}, which is agnostic to the utility structure of the game. To do so, the challenge is that the delay functions in the perturbed game are considerably more convoluted, so much so that, in fact, the perturbed game may no longer be a congestion game! The basic reason for this is that, in the perturbed game, the delay function can depend not just on the number of players using it, but also on their identities. Yet, we observe that this cannot happen when players have the same strategy set, paving the way to~\Cref{theorem:polysteps}. We defer the formal proof of this result along with some relevant background on matroid congestion games to \Cref{sec:matroid}.

The second class of games we consider is \emph{symmetric network congestion} games. Here, players' strategies correspond to paths linking a designated source to its destination. For this class of games, \citet{Fabrikant04:Complexity} gave a polynomial-time algorithm for finding a global optimum of the potential function---and thereby a pure Nash equilibrium---based on min-cost flow. We provide an analogous reduction mapping perfect equilibria to solving a \emph{symbolic} min-cost flow problem, which in turn can be solved using known combinatorial algorithms or linear programming. We prove \Cref{theorem:symmetricnetwork} in \Cref{sec:network}.
\begin{theorem}
    \label{theorem:symmetricnetwork}
    In symmetric network congestion games, there is a polynomial-time algorithm for finding a perfect equilibrium.
\end{theorem}

\paragraph{Strongly polynomial-time algorithms and perturbed optimization} A recurring question that arises throughout the line of work on equilibrium refinements---and our paper in particular---is whether an algorithm can be made to run efficiently under a symbolic class of instances. What we have been implicitly using so far is that certain basic operations---such as comparisons and additions---can still be performed symbolically when $\epsilon$ is small enough (\Cref{lemma:interpolation,lem:integer potential}). We make a significantly more general connection: any \emph{strongly polynomial-time algorithm}---which comprises basic arithmetic operations, namely $\{+, *, /, < \}$---can be made to run \textit{symbolically}; further background and definitions is provided in~\Cref{sec:strongly-poly}.

\begin{restatable}{theorem}{strongpoly}
    \label{theorem:stronglypolynomial}
    Let $\vx \in \Q[\epsilon]^M$ be a vector each of whose entries is a polynomial in $\epsilon$ of degree at most $d_x$. Suppose further that the output of each gate of $\cA$ with input $\vx$ is a piecewise rational function such that each piece has degree at most $d_{\cA} = \poly(d_x, |\cG|)$. Then there is a polynomial-time algorithm that computes the output of the circuit symbolically for any sufficiently small $\epsilon > 0$.
\end{restatable}

Most of our positive results---with the one exception of linear programming---can be seen as instances of this more general connection. The basic idea is to execute the underlying algorithm on multiple inputs in the neighborhood of $\epsilon \approx 0$, and then use fractional interpolation to derive the exact form of the symbolic output. On the other hand, weakly polynomial-time algorithms do not share this property; the simplest example where this becomes evident is binary search over the real numbers.

\subsection{Best-response paths: exponential separations}

The landscape that has emerged so far is that, in potential games, computing pure perfect (and proper) equilibria is polynomial-time equivalent to computing pure Nash equilibria. On the other hand, our next result shows that, for certain potential games, the length of the symbolic best-response path in the perturbed game can be exponentially larger than the length of the best-response path in the original game. (The formal proofs from this subsection are in \Cref{sec:longpaths}.)

\begin{theorem}
    \label{theorem:expsep}
    There is a class of identical-interest games with the following property:
    \begin{itemize}
        \item From any starting point, best-response dynamics converge to a Nash equilibrium in polynomially many steps, whereas
        \item there exist starting points such that symbolic perfect (or proper) best-response dynamics can take exponentially many steps.
    \end{itemize}
\end{theorem}

This class of games is predicated on the usual reduction from the local version of \textsc{MaxCut} with respect to the \textsc{Flip} neighborhood, but with a crucial twist: every node is to be represented by a triplet of players, each of whom has a say on which side of the cut the corresponding node should belong to. In particular, the assignment of that node is decided by the majority. In addition, we create a small incentive for the triplets to be unanimous; since the underlying game needs to be identical-interest, this added term in the utility accounts for the cumulative number of unanimous triplets. A subtle issue in our construction, whose role will become clear shortly, is that the penalty for non-unanimous triplets is only present when there are at least $2$ non-unanimous triplets. The upshot is that best-response dynamics quickly get trapped into spurious equilibria in which all but at most $1$ triplet are unanimous; these are spurious in the sense of failing to account for the local structure of \textsc{MaxCut}. On the other hand, \emph{symbolic} best-response dynamics can only converge to local optima of \textsc{MaxCut}.

In more detail, when we execute best-response dynamics and there are at least $2$ non-unanimous triplets, at some point we will update a player whose action differs from that of the majority in that triplet. That player's best response will be to switch to the side of the majority, for there is an incentive to maximize the number of unanimous triplets. From then onward, players in that newly unanimous triplet will retain their action under best-response dynamics given that the side of the cut is decided by the majority. As a result, we will end up with at most $1$ non-unanimous triplet, and such a triplet can perform a best-response update only once.

We now turn to the second part of~\Cref{theorem:expsep}. Following the above reasoning, we can assume that there is at most $1$ non-unanimous triplet. Thereupon, we claim that after the players in that triplet perform an update under symbolic best-response dynamics, there will be unanimity, and what is more, the elected node will be on the side of the cut that (at that point in time) locally maximizes the weight. In fact, the key claim is that those two invariances will be maintained throughout, which means that symbolic best-response dynamics can only converge to local optima of $\textsc{MaxCut}$ under the $\textsc{Flip}$ neighborhood. The proof proceeds as follows. Let us consider a unanimous triplet in which the elected node is not on the optimal side of the partition. Under symbolic best-response dynamics, where the player accounts for other players' trembles, there are two conflicting forces \emph{neither of which is present under best-response dynamics}: choosing the other side of the partition will incur a penalty since it breaks unanimity, but will also result in switching the side of the node. The crucial point here is that both of these forces are \emph{first-order}, in the sense that they both manifest in the coefficient of the degree-$1$ term in the polynomial that measures progress; this is where we make use of the fact that the penalty for non-unanimity is only introduced in the presence of at least $2$ non-unanimous triplets, so that it only introduces a higher order effect. The claim then follows by selecting the penalty term to be sufficiently small, so that the local structure of $\textsc{MaxCut}$ outweighs the cost of breaking unanimity.

Our next result turns~\Cref{theorem:expsep} on its head: it shows that, for certain potential games, symbolic best-response dynamics can be fast while best-response dynamics can be slow.

\begin{theorem}
    \label{theorem:expsep-opo}
    There is a class of identical-interest games with the following property:
    \begin{itemize}
        \item From any starting point, symbolic perfect (resp.\ proper) best-response dynamics converge to a perfect (resp.\ proper) equilibrium in polynomially many steps, whereas
        \item there exist starting points such that best-response dynamics can take exponentially many steps.
    \end{itemize}
\end{theorem}

As before, we start from the usual reduction from $\textsc{MaxCut}$, but we now introduce two additional players, each with two actions---say $\textsf{d}$ and $\textsf{e}$. When either of the two additional players selects $\textsf{d}$, the game proceeds as usual. But when $(\textsf{e}, \textsf{e})$ is selected, every player gets to collect a large reward $M \gg 1$ plus a bonus proportional to the number of players that belong to the left partition. Now, if the two additional players start from $(\textsf{d}, \textsf{d})$, it follows that best-response dynamics can only converge to local optima of $\textsc{MaxCut}$. This is not so under symbolic best-response dynamics: the possibility of a tremble immediately entices both additional players to action \textsf{e}, whereupon all players will end up entering the left partition---this is so because of the bonus given under $(\textsf{e}, \textsf{e})$.

\Cref{theorem:expsep-opo} is not a computational result; to be sure, finding a Nash equilibrium is no harder than finding a perfect (or proper) equilibrium, but it does present a counter-intuitive exponential separation under a natural class of algorithms.

\subsection{Complexity results for mixed strategies} The foregoing results fully characterize the complexity of computing a \emph{pure} equilibrium refinements in potential games under different game representations. We now turn our attention to mixed strategies. 

We begin with the class of polymatrix (multi-player) games (\Cref{def:polymatrix}). Combining~\Cref{theorem:PLS-complete} with the known $\PPAD$ inclusion due to~\citet{Hansen18:Computational},\footnote{\citet{Hansen18:Computational} showed inclusion for computing \emph{$\epsilon$-symbolic} proper equilibria per the perturbation of~\citet{Kohlberg86:strategic}; one can verify in polynomial time that any such purported solution is indeed an $\epsilon$-symbolic proper equilibrium.} together with the fact that $\CLS = \PPAD \cap \PLS$~\citep{Fearnley23:Complexity}, we arrive at the following consequence.

\begin{corollary}
    \label{cor:CLS}
    Computing a perfect or proper equilibrium of a potential, (multi-player) polymatrix game in normal form is in $\CLS$.
\end{corollary}

It is generally believed that even computing a Nash equilibrium of an identical-interest polymatrix game is $\CLS$-hard; for example, \citet{Hollender25:Complexity} recently showed $\CLS$-hardness when the team is facing multiple independent adversaries. If that conjecture turns out to hold, \Cref{cor:CLS} would imply $\CLS$-completeness for both perfect and proper equilibria. Be that as it may, there is a precise sense in which perfect and proper equilibria in potential polymatrix games are polynomial-time equivalent to Nash equilibria: it can be shown that for a sufficiently small value of $\epsilon$, representable in polynomially many bits, computing a perfect or proper equilibrium reduces to finding a Nash equilibrium of a perturbed potential game (\Cref{prop:polymatrix-round}). We note that this argument does not go through beyond polymatrix games~\citep{Etessami14:Complexity}, and that more generally, the proof of~\Cref{prop:polymatrix-round} makes use of the general theory of~\citet{Farina17:Extensive} pertaining to perturbed LCPs.

\Cref{prop:polymatrix-round} implies that a perfect equilibrium can be computed by running a perturbed variant of gradient descent, thereby furnishing a more direct $\CLS$ membership proof. Interestingly, we observe that the \emph{symbolic} version of gradient descent can get stuck (\Cref{prop:symbolicfails}), so the previous claim only holds by setting $\epsilon$ to a sufficiently small numerical value, representable with polynomially many bits. 

At this point, it is worth revisiting the example of~\Cref{fig:nash-vs-perfect} (\Cref{sec:eq-ref}). An unsatisfactory feature of that example is that while a nonsensible Nash equilibrium exists, it is unlikely to be reached by natural learning algorithms. In particular, in that example, if one initializes gradient descent at random---a common practice in optimization with strong theoretical guarantees in the unconstrained setting (\Cref{sec:related})---the dynamics converge to the optimal equilibrium \emph{almost surely}. But we observe that this is not always the case even in $2 \times 2$ potential games. 

An example game is given in~\Cref{fig:randomGD}. The corresponding gradient descent dynamics are illustrated in~\Cref{fig:quiver_perfect}. We see that vanilla gradient descent often converges to suboptimal Nash equilibria. On the other hand, executing the dynamics on the perturbed game converges to the welfare-optimal point. From a price of anarchy perspective, the example of~\Cref{fig:nash-vs-perfect} already shows that perfect equilibria can be arbitrarily better than Nash equilibria. \Cref{fig:quiver_perfect} takes a step further: it reveals that the \emph{average} price of anarchy of gradient descent---in the sense of~\citet{Sakos24:Beating}---can still be far from the welfare attained by perfect equilibria. On the flip side, we will later discuss how to use~\Cref{theorem:PLS-complete} in conjunction with the framework of~\citet{Roughgarden14:Barriers} to obtain non-trivial lower bounds for the price of anarchy of perfect equilibria.

\begin{figure}
\centering
\begin{tabular}{c c c}
& \textsf{C1} & \textsf{C2} \\
\textsf{R1} & 12, 2 & 2, 2 \\
\textsf{R2} & 11, 1 & 0, 0 \\
\end{tabular}
\caption{A $2 \times 2$ potential game in normal form, the column player's utility function defines the potential. Any strategy profile in which \textsf{R1} is played with probability $1$ is a Nash equilibrium. The only perfect equilibrium is $(\textsf{R1}, \textsf{C1})$. }
\label{fig:randomGD}
\end{figure}

\begin{figure}[!ht]
    \centering
    \includegraphics[scale=0.6]{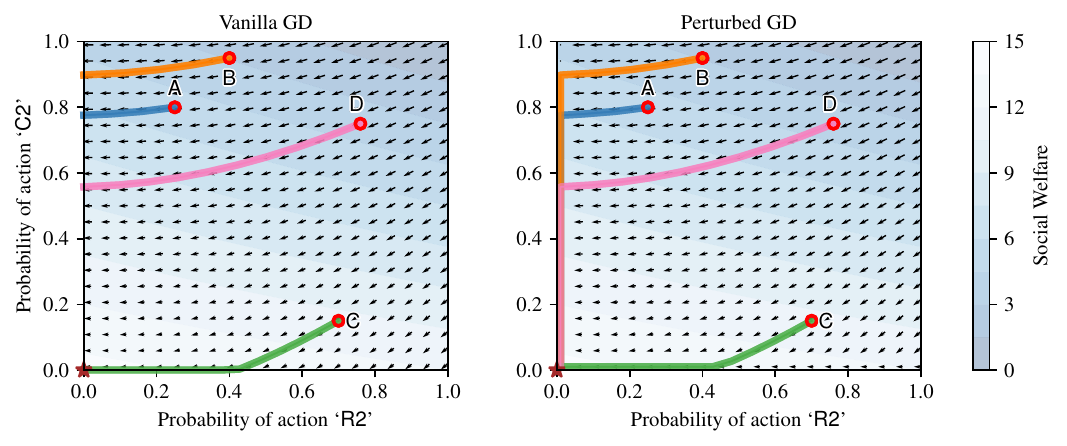}
    \caption{Vanilla gradient descent (left) on the game of~\Cref{fig:randomGD} versus gradient descent on the perturbed game with perturbation magnitude $\epsilon \defeq 0.01$ (right) under $4$ different initializations. The welfare-optimal point and unique perfect equilibrium $(\textsf{R1}, \textsf{C1})$ appears in the bottom-left corner. We see that gradient descent in the unperturbed game often converges to Nash equilibria with much lower welfare than the optimal one.}
    \label{fig:quiver_perfect}
\end{figure}

\paragraph{Beyond polymatrix games} The landscape changes dramatically when moving to general potential games. Indeed, it is not even known whether computing an exact mixed Nash equilibrium in potential games---even though one supported on rational numbers always exists---lies in $\PPAD$, let alone any of its refinements. Instead, based on existing results, we can only place mixed Nash equilibria in $\PLS \cap \FIXP$. Characterizing the exact complexity of this problem is an important open problem in this area; to our knowledge, even the relation between $\FIXP$ and $\PLS$ is unexplored. So what can we hope to prove for the complexity of equilibrium refinements?

\citet{Etessami14:Complexity} have shown that, in general-sum games, as long as one takes the perturbation parameter $\epsilon$ to be \emph{doubly exponentially small}, one immediately recovers a point that is close geometrically to an exact perfect equilibrium. Now, if one posits general arithmetic circuits, one can produce doubly exponentially small values through repeated squaring, thereby reducing perfect equilibria to Nash equilibria. But, of course, this is not possible in the standard model of computation. The main question here is whether one \emph{needs} to take $\epsilon$ doubly exponentially small in potential games; after all, such games are very structured. Our next main result shows that it is indeed necessary.

\begin{theorem}
    \label{theorem:doublyexponentiallysmall}
    For every positive integer $n$ there exists a normal-form potential game $\Gamma_n$ with $4n + 1$ players and two actions per player such that, for all $\eps \in [1/2^{2^n}, 1/2]$, the perturbed game $\Gamma^{(\eps)}_n$ admits a Nash equilibrium that is distance $1/2$ away in $\ell_\infty$-norm from any Nash equilibrium of $\Gamma_n$.
\end{theorem}

As an immediate byproduct of independent interest, this theorem shows that, in potential games, even when $\epsilon$ is doubly exponentially small, an $\epsilon$-approximate Nash can still be far from an exact Nash in geometric distance. (This holds because an exact Nash equilibrium in the perturbed game induces an approximate Nash in the original game with approximation proportional to the perturbation.) This was hitherto only known in three-player general-sum games due to the tour de force of~\citet{Etessami07:Complexity}. Besides holding for a significantly more structured class of games, \Cref{theorem:doublyexponentiallysmall} is established based on a simple and succinct argument; this is unlike the original proof of~\citet{Etessami07:Complexity}, which is especially intricate.

We now sketch the main idea of our argument. It is instructive to begin with general polynomial optimization problems---without restricting to multilinear polynomials. The starting observation is that one can perform repeated squaring by considering the function $f : [0, 1]^{n} \to \R$ given by
\begin{align*}
    f(\vx) = \qty(2x_1 - 1)^2 + \qty(x_2 - x_1^2)^2 + \dots + (x_{n-1} - x_{n-2}^2)^2 - x_n x_{n-1},
\end{align*}
which is to be minimized. It is not hard to show that the point $\vx = (1/2, 1/4, 1/16, \dots, 1/2^{2^{n-2}}, 0)$ is an $\eps$-KKT point with $\eps$ being proportional to $1/2^{2^{n-2}}$. However, every exact KKT point has $x_{n-1} > 0$ and therefore $x_n = 1$. The name of the game now is to make this argument go through while using only multilinear polynomials. To do so, we construct the potential function
\begin{align*}
        \Phi(\vx, \vx', \vc, \vd, t) = \sum_{i=1}^n \qty[\qty(t - c_i) (x_i - x_{i-1} x_{i-1}') + \qty(d_i - \frac12) (x_i - x_i')] - 2x_n - 2x'_n - 2n \cdot t,
\end{align*}
where $\vx, \vx', \vc, \vd \in [0, 1]^n$ and $t \in [0, 1]$; for simplicity of notation, we set $x_0 = x_0' \defeq 1/2$. The sole purpose of the player $t$ is to have a value that is guaranteed to be $\eps$, since by construction $t$ always has negative gradient.  The basic idea is that $\vx'$ is intended to be a copy of $\vx$, and it is essential to make sure that, in equilibrium, $\vx \approx \vx'$. The argument boils down to proving the following two claims.

\begin{itemize}
    \item In every (exact) Nash equilibrium of $\Gamma_n$, there is some $i \in [n]$ such that $d_i \in \{0, 1\}$ or $c_i = 1$.
    \item For $\epsilon \in [1/2^{2^{n}}, 1/2]$, the perturbed game $\Gamma_n^{(\epsilon)}$ has an equilibrium in which $\vec{d} = \frac{1}{2} \vec{1}$ and $\vec{c} = \epsilon \vec{1}$.
\end{itemize}

We provide the details of the argument in~\Cref{sec:doublyexpo}. Also, we show how to embed this construction even in a $3$-player potential game in normal form (\Cref{theorem:3playerdoublyexp}).

\subsection{Price of anarchy of perfect and proper equilibria}

We have seen through some simple motivating examples in~\Cref{sec:eq-ref} that perfect equilibria can lead to significantly higher welfare \emph{vis-\`a-vis} Nash equilibria. On the flip side, we also provide lower bounds for the price of anarchy with respect to perfect and proper equilibria by leveraging the previously established $\PLS$ membership (\Cref{theorem:PLS-complete}). This makes use of the elegant framework of~\citet{Roughgarden14:Barriers} that relies on hardness of approximation for the underlying optimization problem. We provide a concrete application in polynomial congestion games---perhaps the most well-studied class from a price of anarchy perspective---by leveraging the hardness result of~\citet{Paccagnan24:Congestion}.

\begin{theorem}[Precise version in~\Cref{theorem:poa-lower}]
    \label{theorem:informal-poa-lower}
    In polynomial degree-$d$ congestion games, the price of anarchy of pure perfect or pure proper equilibria is at least $d^{\Omega(d)}$.
\end{theorem}
\section{Further related work}
\label{sec:related}

We conclude by expanding on additional related work. We begin by covering many of the advances in the complexity of equilibrium refinements. We then connect our work with prior literature on improving the quality of equilibria in potential games and constrained optimization problems more broadly.

\paragraph{Complexity of equilibrium refinements} \citet{Hansen10:Computational} proved that verifying whether a strategy profile is a perfect equilibrium even in $3$-player games is \NP-hard (and \sqrtsum-hard); this of course stands in stark contrast to (exact) Nash equilibria that trivially admit polynomial-time verifiers. (For two-player games, it is known \citep{VanDamme91:Stability} that a strategy profile is a perfect equilibrium if and only if it is undominated, which can be in turn ascertained in polynomial time via linear programming.)

\citet{Etessami14:Complexity} showed that approximating a perfect equilibrium---in the sense of being close in $\ell_\infty$ distance to an exact one, forming a ``strong approximation'' guarantee in the parlance of~\citet{Etessami07:Complexity}---in games with at least three players is polynomial-time equivalent to approximating Nash equilibria, thereby being $\FIXP_a$-complete. $\FIXP_a$ is a complexity class introduced by~\citet{Etessami07:Complexity} that contains search problems reducible to (strongly) approximating a Brouwer fixed point of a function given by an algebraic circuit with gates $+, -, *, /, \max, \min$. On the other hand, computing an exact perfect equilibrium in two-player games is \PPAD-complete; this follows by carefully analyzing exact pivoting algorithms~\citep{Stengel02:Computing}.

\citet{Etessami21:Complexity} characterized the complexity of various refinements in multi-player extensive-form games of perfect recall. Namely, \emph{sequential equilibria (SEs)}~\citep{Kreps82:Sequential},\footnote{The definition of a sequential equilibrium is somewhat cumbersome, being predicated on a set of beliefs; we do not expand on that definition since it is not relevant for our purposes.} which refine both Nash equilibria and \emph{subgame-perfect equilibria (SPEs)}; extensive-form perfect equilibria (EFPEs), which refine SEs; normal-form perfect equilibria (NFPEs); and quasi-perfect equilibria (QPEs)~\citep{vanDamme84:relation}. The most refined of those notions are QPEs and EFPEs.

\citet{Hansen18:Computational} examined the complexity of proper equilibria. They showed that even in two-player games in normal form, the task of verifying the proper equilibrium conditions is \NP-complete; this is in contrast to perfect equilibria~\citep{Hansen10:Computational}. For multi-player games, they showed that strongly approximating a proper equilibrium is $\FIXP_a$-complete, while computing a \emph{symbolic} proper equilibrium in polymatrix games is \PPAD-complete; the latter strengthens an earlier result due to~\citet{Sorensen12:Computing}. Relatedly, \citet{Hansen21:Computational} studied the complexity of \emph{quasi-proper equilibria} in extensive-form games, a refinement of quasi-perfect equilibria~\citep{vanDamme84:relation}.

A general technique for proving membership in $\FIXP$ was recently developed by~\citet{FilosRatsikas21}; among others, they provided a simple proof handling proper equilibria. Even more recently, \citet{FilosRatsikas24:PPAD} introduced a simple framework for proving $\PPAD$ membership, and gave several new results on the complexity of equilibrium refinements.

It should be noted that the complexity landscape is drastically different in (two-player) zero-sum games. For example, \citet{Miltersen06:Computing} showed that an exact proper equilibrium can be computed in polynomial time through linear programming, and similar positive results have been established for other equilibrium refinements as well. Notably, \citet{Miltersen08:Fast} showed that a (normal-form) proper equilibrium can be computed in polynomial time. A more practical algorithm for zero-sum games was developed by~\citet{Farina18:Practical}. \citet{Bernasconi24:Learning} developed learning algorithms for EFPEs in zero-sum games.

\paragraph{Improving the quality of equilibria in potential games} Our work can also been viewed as part of the research agenda endeavoring to identify better equilibria in potential games~\citep{Anshelevich08:Price,Balcan13:Circumventing,Gemp22:D3C,Sakos24:Beating}, thereby circumventing the \emph{price of anarchy}---the ratio between the welfare-optimal state and the worst-case welfare attained at a Nash equilibrium. Yet the viewpoint we take in this paper is quite different. One related work by~\citet{Balcan18:Diversified} shows that \emph{diversified} strategies---ones that do not put too much probability mass on a single action---can lead to higher-welfare equilibria; the idea of diversification closely ties to the perfect equilibrium perturbation. Another related concept is ``price of uncertainty,'' introduced by~\citet{Balcan13:Price}. It deals with the impact of small fluctuations in the behavior of best response dynamics. Such fluctuations can be due to players' mistakes, not unlike perfect equilibria. 
 
As we have discussed already, one straightforward observation---which follows readily even from the simple example of~\Cref{fig:nash-vs-perfect}---is that the price of anarchy defined with respect to perfect equilibria can be arbitrarily smaller than that with respect to Nash equilibria. These type of considerations are not without precedent: \citet{Leme12:Curse} introduced the ``sequential price of anarchy.'' Informally, it relates the outcome of the worst possible \emph{subgame perfect equilibrium} of all sequential versions of the game. They showed that the sequential price of anarchy can be significantly more favorable in certain classes of games.

\paragraph{Optimization perspective} On a similar note, there has been tremendous interest on understanding how to avoid certain undesirable stationary points in constrained optimization problems. In unconstrained problems, it is known that a perturbed version of gradient descent converges to second-order stationary points, thereby avoiding \emph{strict} saddle points~\citep{Jin17:How, Lee19:First}. The situation is more complex in the constrained setting. \citet{Nouiehed18:Convergence} showed that gradient descent with random initialization can often fail to converge to second-order stationary points; in fact, they showed that even checking the conditions of second-order stationarity is \NP-hard. Our work contributes to this line of work by identifying a strict subset of stationary points that can be reached through a perturbed variant of gradient descent. It is worth referring to~\citet{vanderLann99:Existence} for a notion of ``perfect stationary point.'' In particular, \citet{Dang15:Interior} came up with an interior-point-type algorithm for finding such refinements.
\section{Conclusions and future research}

Our main contribution in this paper was to characterize the complexity of equilibrium refinements in potential games. Our results span various game representations---ranging from normal-form games to general polytope games---and equilibrium refinements, primarily centered on (normal-form) perfect and proper equilibria. Taken together, our results paint a comprehensive picture of the complexity landscape.

Moreover, some of our technical results have important implications well beyond equilibrium refinements. Most notably, we significantly strengthened and simplified a seminal result of~\citet{Etessami07:Complexity}: even in three-player potential games, we showed that a doubly exponentially small precision is necessary to be geometrically close to an equilibrium. 

There are several interesting avenues for future research. First, the complexity of computing a strong approximation to a \emph{mixed} Nash equilibrium in potential games remains an outstanding open problem. 

\begin{question}
    \label{question:FIXP}
    Is computing an exact mixed Nash equilibrium in potential games $\FIXP \cap \PLS$-complete? Is the strong approximation thereof $\FIXP_a \cap \PLS$-complete.
\end{question}

It is clearly in both $\FIXP$ and $\PLS$,\footnote{This is so when one adopts the definition of $\PLS$ we use in this paper, which does not rest on being able to verify any possible purported solution; otherwise, verifying whether a point is close to an exact mixed Nash equilibrium is a hard problem.} but nothing is known beyond $\CLS$-hardness. In light of our results, a $\FIXP \cap \PLS$-hardness result would immediately resolve the complexity of both perfect and proper equilibria. To our knowledge, the relation between $\FIXP$ and $\PLS$ is entirely unexplored; what is the right analog of $\FIXP$ for problems in $\CLS$?

It would also be interesting to expand the positive results we obtained in~\Cref{sec:positive}. For example, what can be said about matroid congestion games without the symmetry assumption? And can those results for perfect equilibria be extended to proper equilibria?

Finally, are there further implications of the connection we made between strongly polynomial-time algorithms and perturbed optimization (\Cref{theorem:stronglypolynomial})? We believe there is fertile ground in employing our techniques to other optimization problems that do not necessarily have a game-theoretic flavor. Our framework can be used to refine the solution quality in problems such as shortest paths and maximum flow by accounting for some uncertainty regarding the underlying costs. Naturally, there has been much work in that direction, but the viewpoint stemming from equilibrium refinements seems to be new.

\section*{Acknowledgments}

We are indebted to Ratip Emin Berker for numerous insightful discussions throughout this project. I.A. thanks Ioannis Panageas, Jingming Yan, and Alexandros Hollender for many discussions pertaining to~\Cref{question:FIXP}. I.A. also thanks Gabriele Farina for a helpful discussion concerning~\Cref{lemma:NPP}. K.F. thanks Vince Conitzer, whose course on ``Foundations of Cooperative AI" inspired the initial ideas that led to this work. T.S. is supported by the Vannevar Bush Faculty Fellowship ONR N00014-23-1-2876, National Science Foundation grants RI-2312342 and RI-1901403, ARO award W911NF2210266, and NIH award A240108S001. E. T. thanks the Cooperative Al Foundation and PhD Fellowship, Macroscopic Ventures and Jaan Tallinn’s donor-advised fund at Founders Pledge for financial support.

\bibliographystyle{plainnat}
\bibliography{refs}

\clearpage

\appendix


\clearpage

\section{Lower bound on verification of perfect equilibrium}
\label{sec:lowerboundverif}

The first question is whether a purported perfect equilibrium in a potential game can be verified efficiently. Efficient verification is obvious for Nash equilibria, but, surprisingly, \citet{Hansen10:Computational} showed that even in three-player general-sum games the problem becomes \NP-hard. In this section, we strengthen the result of~\citet{Hansen10:Computational} by showing that even in a three-player \emph{identical-interest} games \NP-hardness persists.

\begin{theorem}\label{th:verify}
    It is \NP-hard to check whether a given strategy profile is a perfect equilibrium, even when the game is an identical-interest game with three players and the given profile is pure.
\end{theorem}
We dedicate the rest of this section to proving the theorem. \citet{Hansen10:Computational} showed that this problem is \NP-hard for general-sum games; we adapt their proof. We reduce from the problem of computing the team minimax equilibrium in a three-player team game:
\begin{definition}
    An {\em adversarial team game} is an $n$-player game in which $n-1$ players (the ``team'') have the same utility function $u$, and the remaining player (the ``adversary'') has utility $-u$. The {\em team minimax equilibrium (TME) value} is the smallest number $r$ such that the team can force the adversary's utility to be at most $r$ by playing an uncorrelated mixed strategy profile.\footnote{Following \citet{Hansen10:Computational}, values will be from the perspective of the adversary.}  
\end{definition}
    
Computing even approximately the TME value of an adversarial team game is \NP-hard.
\begin{theorem}[\cite{Borgs08:Myth}]
    It is \NP-hard to approximate the TME value of a three-player adversarial team game with $m$ actions per player to within an additive $\Omega(1/m^2)$. That is, given a game and a value $r$, it is \NP-hard to distinguish whether the TME value is $\ge r+\delta$ or $\le r-\delta$, where $\delta = \Omega(1/m^2)$.
\end{theorem}

We are now ready to prove \Cref{th:verify}. Given a three-player team game and a purported value $r$ for the team, construct the following identical-interest game: add an action $\bot$ for each player. Utilities are defined as follows.
\begin{enumerate}
    \item If no one plays $\bot$, everyone gets the value that the {\em adversary} would have gotten in the team game. 
    \item Otherwise, if the adversary and {\em exactly} one team member plays $\bot$, then the utility is $r - \delta/2$.
    \item Otherwise, the utility is $r$.
\end{enumerate}We now prove both directions of the reduction.
\begin{lemma}
    If the TME value is $\le r - \delta$, then $\vx^\bot := (\bot, \bot, \bot)$ is a perfect equilibrium.
\end{lemma}
\begin{proof}
    Let $\vx^*$ be any strategy profile in which the team plays according to the TME. Consider the trembling sequence 
    \begin{align*}
        \vx^{(\eps)} \defeq (1 - \eps - \eps^2) \cdot \vx^\bot + \eps \cdot \vx^* + \eps^2 \cdot \vf
    \end{align*}
    where $\vf$ is any fully mixed strategy profile. Clearly, $\vx^{(\eps)} \to \vx^\bot$ as $\eps \to 0$. For team members, $\bot$ is a strict best response to $\vx^\bot$ (since any other action has expected utility roughly $r - \delta/2$), so team members are best responding in the limit $\eps \to 0$. For the adversary, the value of playing $\bot$ is $r \pm O(\eps^2)$, and the value of playing any other action is upper-bounded by $(1 - \eps) r + \eps (r - \delta) \pm O(\eps^2) = r - \eps \delta \pm O(\eps^2) < r \pm O(\eps^2)$ for $\eps$ sufficiently small, so $\bot$ is a best response. Thus, $\vx^\bot$ is a perfect equilibrium.
\end{proof}
\begin{lemma}
    If the TME value is $\ge r + \delta$, then $\vx^\bot$ is {\em not} a perfect equilibrium.
\end{lemma}
\begin{proof}
    Consider any fully mixed strategy profile $\vx$. Let $\vy$ be the team's strategy in $\vx$, conditioned on playing actions other than $\bot$. We claim that the adversary's best response to $\vx$ cannot be $\bot$, which would complete the proof since then, the adversary cannot play $\bot$ with probability larger than $\eps$ in any $\eps$-PE. Indeed, the expected value of action $\bot$ for the adversary is at most $r$, but best-responding to $\vy$ gets expected value strictly larger than $r$, since it gets $r$ when at least one team member plays $\bot$, and at least $r + \delta$ otherwise.
\end{proof}
\Cref{th:verify} now follows immediately from the previous two lemmas.
\section{Computing pure perfect equilibria}
\label{sec:nfpe}

Notwithstanding~\Cref{th:verify}, this section shows that computing pure perfect equilibria is $\PLS$-complete even in general polytope games. We begin by examining the simple case of explicitly represented normal-form games (\Cref{sec:constantnum}), where we think of the number of players as being a constant. In \Cref{sec:basicPLS} we then deal with concise potential games, and~\Cref{sec:polytope} generalizes our approach to any polytope game.

\subsection{Constant number of players}
\label{sec:constantnum}

To begin with, we deal with perfect equilibria in potential games with a constant number of players. Perhaps the most natural approach is to identify any strategy corresponding to a maximum entry in the payoff tensor. However, this approach falls short: the example of~\Cref{fig:maxpayoff} shows a game in which $3$ different action profiles maximize the utility, but only one of them is a perfect equilibrium: if $\vx_1 \in \Delta(2) \defeq \{ \vx_1 \in \R^2_{\geq 0} : \vx_1(\textsf{R1}) + \vx_1(\textsf{R2}) = 1 \} $ and $\vx_2 \in \Delta(2)$ are fully mixed, then $u_1(\textsf{R1}, \vx_2) > u_1(\textsf{R2}, \vx_2)$ and $u_2(\textsf{C1}, \vx_1) > u_2(\textsf{C2}, \vx_2)$, so $\vx_1(\textsf{R2}) \leq \epsilon$ and $\vx_2(\textsf{C2}) \leq \epsilon$ in the perturbed game. This means that $(\textsf{R1}, \textsf{C1})$ is the unique perfect equilibrium.

\begin{figure}[!h]
\centering
\begin{tabular}{c c c}
& \textsf{C1} & \textsf{C2} \\
\textsf{R1} & 1, 1 & 1, 1 \\
\textsf{R2} & 1, 1 & 0, 0 \\
\end{tabular}
\caption{A $2 \times 2$ identical-interest game in normal form with a unique perfect equilibrium, namely (\textsf{R1}, \textsf{C1}).}
\label{fig:maxpayoff}
\end{figure}

Nevertheless, there is a simple algorithm that solves this problem, which is to compute the maximum entry of the potential function \emph{in the perturbed game} in the limit where $\epsilon$ is arbitrarily small. In the example of~\Cref{fig:maxpayoff}, we have $(\textsf{R1}, \textsf{C1}) \mapsto 1 - \epsilon^2$; $(\textsf{R1}, \textsf{C2}) \mapsto 1 - \epsilon ( 1 - \epsilon)$; $(\textsf{R2}, \textsf{C1}) \mapsto 1 - \epsilon ( 1 - \epsilon)$; $(\textsf{R2}, \textsf{C2}) \mapsto 1 - (1 - \epsilon)^2$. As a result, $(\textsf{R1}, \textsf{C1})$ is to be returned. It is worth pointing out that this approach contains a simple but useful idea: one can transfer the perturbation from the strategy sets into the utilities. We will also make use of this transformation later on.

\begin{proposition}
    \label{prop:constantplayers}
    For any potential game with a constant number of players, there is a polynomial-time algorithm for computing a perfect equilibrium.
\end{proposition}

For every joint action profile $(a_1, \dots, a_n)$, we compute the expected potential as a polynomial in $\epsilon$ when each player $i \in [n]$ plays $a_i$ with probability $1 - (m_i - 1) \epsilon$ and every other action with probability $\epsilon$. Each such polynomial can be written explicitly in polynomial time since the game is assumed to have a constant number of players. We then compute the largest such polynomial, where comparison of polynomials is done lexicographically. It is easy to see that the corresponding action profile is a pure perfect equilibrium of the game. (An analogous approach succeeds even under the stronger notion of proper equilibria, introduced in~\Cref{def:proper}.)

Computing equilibrium refinements becomes more interesting in a concise potential games, which is the subject of the rest of this section.

\subsection{Algorithmic approach and complexity implications}
\label{sec:basicPLS}
We consider best-response dynamics on the perturbed game, $\Gamma^{(\epsilon)}$, in which $\epsilon$ is sufficiently small. At each timestep, we pick an arbitrary player, $i \in [n]$ and let them best-respond to the current strategies of the other players. The player will change their strategy if there exists $a_i \in \cA_i$ such that
\[u_i(\vec{a}_i^{(\epsilon)},\vx_{-i}) > u_i(\vx_i,\vx_{-i}), \]
where with $\vec{a}_i^{(\epsilon)}$ we denote the $\epsilon$-pure strategy that assigns maximum probability mass to action $a_i$. 


 In~\Cref{alg:delta-br} we present the symbolic best-response dynamics and in~\Cref{thm:br-converges-to-perfect} we prove that they converge to a Nash equilibrium of the perturbed game. 


\begin{algorithm}[h]
    \caption{$\epsilon$-symbolic best-response dynamics}\label{alg:delta-br}
    \DontPrintSemicolon
    \SetAlgoLined

    \textbf{Input:} $\eps$-perfect or $\eps$-proper perturbed game with strategy space $\cX^{(\eps)}$
    
    Initialize for all $i \in [n], \; \vx_i^{(\epsilon)} \in \cX_i^{(\epsilon)}$\;
    
    \While{$\exists \; i \in [n], i' \in [m_i]  \; \text{such that} \; u_i(\vx_i^{(\epsilon)}, \vx_{-i}^{(\epsilon)}) \lesseps u_i(\va_{i'}^{(\epsilon)}, \vx_{-i}^{(\epsilon)})$}{\label{line:while}
        Set $\vx_i^{(\epsilon)} = \va_{i'}^{(\epsilon)}$ \label{line:br}
    }
    
    \KwRet $\vx^{(\epsilon)}$\;
\end{algorithm}

\begin{theorem}\label{thm:br-converges-to-perfect}
    \Cref{alg:delta-br} returns a perfect equilibrium of $\Gamma$ in finite time, and each iteration can be implemented in polynomial time in the size of the input.
\end{theorem}

\begin{proof}
    We first argue about the per-iteration complexity. \Cref{alg:delta-br} maintains the invariance that, for all $i \in [n]$ and $a_i \in \cA_i$, $\vx_i(a_i) = \epsilon$ or $\vx_i(a_i) = 1 - (m_i-1) \epsilon$. Thus, for concise games, we can determine in polynomial time the polynomials $p(\epsilon) \defeq u_i(\vx_{i}, \vx_{-i})$ and $q(\epsilon) \defeq u_i(\vec{a}_i^{(\epsilon)}, \vx_{-i})$ in terms of $\epsilon$ by \Cref{lem:polynomial circuit}. Both of those polynomials have degree at most $n$. Let $p(\epsilon) = \sum_{i=0}^{n} p_i \epsilon^i$ and $q(\epsilon) = \sum_{i=0}^n q_i \epsilon^i$. To ascertain whether $p(\epsilon) \lesseps q(\epsilon)$, we determine whether $p_i < q_i$ for some $i \in [n] \cup \{0\}$ and $p_{i'} = q_{i'}$ for all $i' < i$.

    We now turn to the proof of convergence. For an $\eps$-pure strategy profile $\vx^{(\epsilon)} \in \cX^{(\eps)}$, the potential function value is a polynomial $p(\eps)$. This potential function increases (with respect to the ordering $\lesseps$) after every best-response step and it can only take on finitely many values, so the algorithm must eventually stop at some profile $\vx^{(\epsilon)}$. At this point, no player has a symbolic best response, so $\vx^{(\eps)}$ is a Nash equilibrium of the perturbed game in the sense that $u_i(\tilde{\vx}_i^{(\eps)}, \vx_{-i}^{(\eps)}) \lesseps u_i(\vx^{(\eps)}) $ for all $i \in [n]$ and $\tilde{\vx}_i^{(\eps)} \in \cX_i^{(\eps)}$. Therefore, there exists $\epsilon_0 > 0$ small enough such that $u_i(\tilde{\vx}^{(\eps)}) < u_i(\vx^{(\eps)})$ for all $\epsilon \in (0, \epsilon_0)$. Therefore, $\{\vx^{(\eps)}\}_{\eps \to 0+}$ defines a sequence of an $\eps$-perfect equilibria of $\Gamma$, the limit point of which is a perfect equilibrium.

\end{proof}



From a complexity standpoint, we prove $\PLS$-membership for the problem of finding perfect equilibria in potential games.



\begin{theorem}\label{thm:perfect-pls}
    Finding a perfect equilibrium of a concise potential game is in \PLS.
\end{theorem}

\begin{proof}
We construct the two required circuits. The feasible set $\cA$, as suggested by the notation, is the set of pure strategy profiles. A pure strategy profile $a \in \cA$ is interpreted as corresponding to the $\eps$-pure strategy $\vx_a = (\vec{a}_1^{(\epsilon)}, \dots, \vec{a}_n^{(\epsilon)}) \in \cX^{(\eps)}$, where $\eps$ is symbolic. Define the value function $V : \cA \to \N$ by $V(a) = \psi(\Phi(\vx_a))$, where $\psi$ is  the function guaranteed by \Cref{lem:integer potential}. (By \Cref{lem:polynomial circuit}, $V$ is efficiently computable.) Define the step function $S$ as a symbolic better response function: given $a \in \cA$, find a player $i$ and an action $a_i' \in \cA_i$ such that $\Phi(\vx_{(a_i', a_{-i})}) \gteps \Phi(\vx_{a})$ and output $(a_i', a_{-i})$; if no such $a_i'$ exists, output $a$ itself. Since every local optimum of $V$ is by construction also a local optimum of $a \mapsto \Phi(\vx_a)$ (with respect to $\lesseps$), it must also be a perfect equilibrium.
\end{proof}

\subsection{Perfect equilibria of polytope games}\label{sec:polytope}

Next, we generalize the results of the previous section beyond normal-form games, to {\em polytope games}. In this setting, the pure action sets $\cA_i$ are points in $\R^d$, and the polytope $\cX_i$ is the convex hull of $\cA_i$. The utility functions $u_i : \cX \to \R$ are, as usual, assumed to be given as multilinear circuits. We will assume that $\cX_i$ is full-dimensional, and indeed we will assume that for each player $i$ we are given a point $\vo_i$ that lies on the strict interior of $\cX$.  

Our assumptions on the set $\cX_i$ will be very general. Namely, we will only assume that we are given a best response (linear optimization) oracle for each player $i$, that is, an efficient Turing machine that, given a rational vector $\vu \in \R^d$, outputs a rational point $\vx_i \in \cA_i$ that maximizes the inner product $\ip{\vu, \vx_i}$. For polytopes, such an oracle can be implemented using typical oracles for convex sets such as a separation or membership oracle~\cite{Grotschel93:Geometric}. We will also assume that the pure action sets $\cA_i$ are ``nice'', in particular, that they have at most exponentially many vertices, each having polynomially bounded representation size. 

Polytope games concisely describe many general and important classes of games, including Bayesian games (where $\cX_i$ is a product of simplices), network congestion games (where $\cX_i$ is the set of flows on a graph), and extensive-form games (where $\cX_i$ is the sequence-form polytope, which we will elaborate on in \Cref{sec:efg}). 

We will be interested in computing perfect equilibria of polytope games. In this section, we will interpret the polytope game as merely a concise representation of an underlying normal-form game in which the action sets are the $\cA_i$s, so we will be interested in perfect equilibria of the underlying normal-form game. For extensive-form games in particular, there are several alternative ways of defining concepts related to perfect equilibria; we elaborate on those in \Cref{sec:efg} as well. 

We now redo the definitions of \Cref{sec:nfpe} in this more general setting.
\begin{definition}[$\eps$-perturbed game for polytope games]
    \label{def:pert-polytope}
    Let $\eps > 0$. The $\eps$-perturbed strategy set $\cX_i^{(\eps)}$ of player $\cX_i$ is defined as
    \begin{align*}
        \cX_i^{(\eps)} := \{ (1 - \eps) \vx_i + \eps \vo_i : \vx_i \in \cX_i\}.
    \end{align*}
    As before, we define $\cX^{(\epsilon)} = \prod_{i=1}^{n} \cX_i^{(\epsilon)}$ and the perturbed game $\Gamma^{(\epsilon)} =(\cX_1^{(\epsilon)}, \ldots, \cX_n^{(\epsilon)}, u_1, \ldots, u_n)$. 
\end{definition}
With these definitions, the remainder of the analysis from the normal-form section follows analogously, and we will not repeat it. We arrive at the following results, whose proofs follow analogously from the corresponding proofs from \Cref{sec:nfpe}:
\begin{theorem}
    \Cref{alg:delta-br} returns a perfect equilibrium in finite time, and each iteration can be implemented in polynomial time in the size of the input.
\end{theorem}

\polytopepls*

Indeed, the same proof shows the following more general result, which we will reference in the remainder of the paper.



\begin{theorem}\label{th:general}
    Let $\{\cX_i^{(\eps)} \}_{i \in [n]}$, where $\eps$ is symbolic, be a collection of strategy sets, one per player. Assume that for every player $i$ there is an efficient algorithm for linear optimization over $\cX_i^{(\eps)}$. Then the problem of computing a symbolic exact Nash equilibrium of a potential game in which each player $i$ has strategy set $\cX_i^{(\eps)}$ is in $\PLS$ for potential games.
\end{theorem}

\section{Mixed perfect equilibria}
\label{sec:mixed}

Having fully characterized the complexity of pure perfect equilibria in potential games, we are concerned in this section with \emph{mixed} perfect equilibria. \Cref{sec:polymatrix} deals with the special case of polymatrix games, while~\Cref{sec:doublyexpo} presents a lower bound in general potential games.

\subsection{Almost implies near: a refinement in polymatrix games}
\label{sec:polymatrix}

We first restrict our attention to polymatrix games, a specific class of succinct games (introduced earlier in~\Cref{def:polymatrix}). For such games, it follows from the result of~\citet{Hansen18:Computational} that computing a perfect equilibrium is in $\PPAD$.\footnote{More precisely, \citet{Hansen18:Computational} showed that computing an $\epsilon$-symbolic proper equilibrium is in $\PPAD$; this implies that computing a perfect equilibrium is also in $\PPAD$ (\Cref{sec:search}).} Combining with the $\PLS$ membership established in~\Cref{thm:perfect-pls}, it follows that computing a perfect equilibrium in polymatrix potential games is in $\CLS$ since $\CLS = \PPAD \cap \PLS$~\citep{Fearnley23:Complexity}.

\begin{corollary}
    \label{cor:CLSmembership}
    Computing a (mixed) perfect equilibrium in polymatrix potential games is in $\CLS$.
\end{corollary}

In terms of hardness results, it is generally believed that computing a Nash equilibrium is $\CLS$-hard; for example, \citet{Hollender25:Complexity} established $\CLS$-hardness but in the presence of multiple independent adversaries; this $\CLS$-hardness result was further strengthened by~\citet{Anagnostides25:Complexity}. In any event, in what follows, we will provide a certain equivalence between Nash and perfect equilibria in polymatrix potential games (\Cref{prop:polymatrix-round}).

The foregoing $\CLS$ membership in~\Cref{cor:CLSmembership} is established rather indirectly: having established $\PLS$ membership in~\Cref{thm:perfect-pls}, we then make use of the existing $\PPAD$ membership together with the fact that $\CLS = \PPAD \cap \PLS$. This begs the question of whether there is a more direct $\CLS$ membership through the use of gradient descent; this is the main subject of this subsection.

\paragraph{Almost implies near in general games} The key step in the analysis is a refinement of the ``almost implies near'' paradigm of~\citet{Anderson86:Almost}. To begin with, we treat general games---without the polymatrix restriction. As observed by~\citet{Etessami14:Complexity}, a direct application of the result of~\citet{Anderson86:Almost} implies that for any game $\Gamma$ and $\delta > 0$, there is an $\epsilon = \epsilon(\delta)$ such that any $\epsilon$-PE of $\Gamma$ is within $\ell_\infty$ distance $\delta$ from an exact PE of $\Gamma$. \citet{Etessami14:Complexity} provided a quantitative bound using techniques from real algebraic geometry~\citep{Basu08:Algorithms}, which is recalled below. (We use the shorthand notation $M = \sum_{i=1}^n m_i$.)

\begin{lemma}[\citep{Etessami14:Complexity}]
    \label{lemma:almost-near}
    For any game $\Gamma$, a sufficiently small $\delta \in \Q_{>0}$, and $\epsilon \leq \delta^{ n^{O(M^3)}}$, any $\epsilon$-PE of $\Gamma$ is within $\ell_\infty$ distance $\delta$ from an exact PE of $\Gamma$.
\end{lemma}

We first point out a refinement of the above lemma that accounts for \emph{$\epsilon$-perfect $\epsilon'$-almost} equilibria of $\Gamma$, which means that $u_i(a_i, \vx_{-i}) < u_i(a_i', \vx_{-i}) - \epsilon' \implies \vx_i(a_i) \leq \epsilon$ for all players $i \in [n]$ and actions $ a_i, a_i' \in \cA_i$; this is the $\epsilon'$-well-supported approximation of $\epsilon$-perfect equilibria, which is meaningful only when $\epsilon' \ll \epsilon$.\footnote{Any Nash equilibrium can be trivially converted into an $\epsilon$-perfect $O_\epsilon(\epsilon)$-almost equilibrium by redistributing the probability mass of each strategy so that each action is allotted a probability of at least $\epsilon$.} We will prove the following property.

\begin{lemma}
    \label{lemma:new-almostimpliesnear}
    For any game $\Gamma$, $\epsilon \in \Q_{> 0}$, $\delta \in \Q_{> 0}$, and $\epsilon' \leq \min(\epsilon, \delta)^{n^{O(M^2)}}$, any $\epsilon$-perfect $\epsilon'$-almost equilibrium of $\Gamma$ is within $\ell_\infty$ distance $\delta$ from an $\epsilon$-PE of $\Gamma$.
\end{lemma}

\begin{proof}
For a \emph{fixed} $\epsilon > 0$, we define $\textsc{AlmostPE}(\vx, \epsilon, \epsilon')$ to be the quantifier-free first-order formula with free variables $\vx \in \R^M$ and $\epsilon' > 0$ expressed as the conjunction of the following formulas:
\begin{equation}
\tag{\textsc{AlmostPE}$(\vx, \epsilon, \epsilon')$}
\begin{aligned}
\vx_i(a_i) > 0 \quad i \in [n], a_i \in \cA_i, \\
\sum_{a_i \in \cA_i} \vx_i(a_i) = 1 \quad i \in [n], \\
(u_i(a_i, \vx_{-i}) \geq u_i(a_i', \vx_{-i}) - \epsilon') \lor (\vx_{i}(a_i) \leq \epsilon) \quad i \in [n], a_i, a_i' \in \cA_i.
\end{aligned}
\label{eq:almostpe-formula}
\end{equation}
We also define $\textsc{PE}(\vx, \epsilon) \equiv \textsc{AlmostPE}(\vx, \epsilon, 0)$ with free variable $\vx \in \R^M$. What we want to prove is that
\begin{equation*}
    \forall \vx \in \R^M \exists \vx' \in \R^M : ( \epsilon' > 0) \land ( \lnot \textsc{AlmostPE}(\vx, \epsilon, \epsilon') \lor ( \textsc{PE}(\vx') \land \|\vx - \vx'\|^2_2 \leq \delta^2 )).
\end{equation*}
($\|\vx - \vx'\|_2 \leq \delta$ implies $\| \vx - \vx' \|_\infty \leq \delta$.) We denote by $\textsc{PEBound}_{\epsilon, \delta}(\epsilon')$ the above first-order formula, with $\epsilon'$ acting as a free variable. Suppose that each payoff can be represented with at most $\tau$ bits, $\epsilon = 2^{-r}$, and $\delta^2 = 2^{-k}$. For $\textsc{PEBound}_{\epsilon, \delta}(\epsilon')$ we have the following facts:
\begin{itemize}
    \item The maximum degree of all involved polynomials in the formula is at most $\max\{2, n - 1\}$.
    \item Each coefficient requires at most $\max(k, \tau, r)$ number of bits to be represented.
    \item There is only a single free variable, namely $\epsilon'$.
    \item The formula is in prenex normal form, containing two blocks of quantifiers of size $M$ and $M$, respectively.
\end{itemize}
Using quantifier elimination~\citep[Algorithm 14.5]{Basu08:Algorithms}, we can convert $\textsc{PEBound}_{\epsilon, \delta}(\epsilon')$ to an equivalent quantifier-free formula $\textsc{PEBound}'_{\epsilon, \delta}(\epsilon')$, again with a single free variable $\epsilon$. It follows, by~\citet[Theorem 14.16]{Basu08:Algorithms}, that the degree of all involved polynomials (which are univariate polynomials in $\epsilon$) is at most $n^{O(M^2)}$, and the bit complexity of each coefficient in the formula is at most $\max(k, r, \tau) n^{O(M^2)}$. By~\citet{Anderson86:Almost}, we know that $\textsc{PEBound}'_{\epsilon, \delta}(\epsilon')$ is satisfied for a sufficiently small $\epsilon' > 0$. It then follows that $\textsc{PEBound}'_{\epsilon, \delta}(\epsilon')$ is satisfied for some $\epsilon' = \epsilon^* \geq 2^{ - \max(k, r, \tau) n^{O(M^2)}}$~\citep{Basu08:Algorithms}, which in turn implies that $\textsc{PEBound}'_{\epsilon, \delta}(\epsilon')$ is satisfied for any $\epsilon' \leq \epsilon^*$ in view of the semantics of the formula.
\end{proof}

\Cref{lemma:new-almostimpliesnear} together with~\Cref{lemma:almost-near} imply that, in principle, one can use gradient descent on the perturbed game $\Gamma^{(\epsilon)}$, for a concrete numerical value $\epsilon$, to find a strategy profile that is arbitrarily close to a perfect equilibrium in any potential game. The problem, of course, is that $\epsilon$ needs to be doubly exponentially small, and the number of iterations of gradient descent will be doubly exponentially large since $\epsilon' \leq \delta^{ n^{O(M^3)} n^{O(M^2)}} = \delta^{n^{O(M^3)}} $ (\Cref{lemma:new-almostimpliesnear}).

More precisely, the previous claim follows from the usual analysis of gradient descent (\Cref{prop:gd-folklore}), which we include below for completeness. 

\paragraph{Analysis of gradient descent} Under a constraint set $\cX$, the update rule of (projected) gradient descent reads
\begin{equation}
    \tag{GD}
    \vx{(t+1)} = \Pi_{\cX} ( \vx(t) + \eta \nabla \Phi(\vx(t))),
\end{equation}
where $\Pi_\cX$ denotes the Euclidean projection. Now, if $\Phi$ is $L$-smooth, meaning that $\| \nabla \Phi(\vx) - \nabla \Phi(\vx') \|_2 \leq L \| \vx - \vx' \|_2$, we have
\begin{equation}
    \label{eq:quadraticbound}
    \Phi(\vx') \geq \Phi(\vx) + \langle \nabla \Phi(\vx), \vx' - \vx \rangle - \frac{L}{2} \|\vx - \vx'\|_2^2
\end{equation}
for any $\vx, \vx' \in \cX$. By the nonexpansiveness of $\Pi_\cX$, we have
\begin{equation}
    \label{eq:closeness}
    \| \vx(t+1) - \vx(t) \|_2 = \| \Pi_{\cX} ( \vx(t) + \eta \nabla \Phi(\vx(t))) - \Pi_\cX(\vx(t)) \|_2 \leq \eta \|\nabla \Phi(\vx(t)) \|_2 \leq \eta L,
\end{equation}
where we assume that $\| \nabla \Phi(\vx) \|_2 \leq L$ for all $\vx \in \cX$. Furthermore, by the first-order optimality condition of gradient descent, we have
\begin{equation*}
    \langle \eta \nabla \Phi(\vx(t)) - (\vx(t+1) - \vx(t)), \vx' - \vx(t) \rangle \geq 0 \quad \forall \vx' \in \cX,
\end{equation*}
which in turn implies
\begin{equation}
    \label{eq:fo-opt}
    \langle \nabla \Phi(\vx(t)), \vx(t+1) - \vx(t) \rangle \geq \frac{1}{\eta} \| \vx(t+1) - \vx(t) \|_2^2.
\end{equation}
Combining~\eqref{eq:quadraticbound}, \eqref{eq:closeness}, and~\eqref{eq:fo-opt}, we conclude that
\begin{equation*}
    \Phi(\vx(t+1)) - \Phi(\vx(t)) \geq \left( \frac{1}{\eta} - \frac{\eta^2 L^3}{2} \right) \| \vx(t+1) - \vx(t) \|_2^2 \geq \frac{1}{2\eta} \| \vx(t+1) - \vx(t) \|_2^2 
\end{equation*}
when $\eta \leq 1/L$. The telescopic summation yields that after $O_\epsilon(1/\epsilon^2)$ iterations there is a point $\vx(t)$ such that $\| \vx(t+1) - \vx(t) \|_2 \leq \epsilon$. By the first-order optimality condition, this implies
\begin{equation*}
    \langle \nabla \Phi(\vx(t)), \vx' - \vx(t) \rangle \geq - \frac{D_\cX}{\eta} \epsilon \quad \forall \vx' \in \cX,
\end{equation*}
where $D_\cX$ denotes the $\ell_2$ diameter of $\cX$. We summarize this standard guarantee below.

\begin{proposition}
    \label{prop:gd-folklore}
    Gradient descent on any $L$-smooth function on a bounded domain takes at most $O_\epsilon(1/\epsilon^2)$ iterations to converge to a first-order stationary point $\vx(t)$, that is,
    \begin{equation*}
        \langle \nabla \Phi(\vx(t)), \vx' - \vx(t) \rangle \geq - \epsilon \quad \forall \vx' \in \cX,
    \end{equation*}
\end{proposition}

In our application, we set the constraint set to be the $\epsilon$-perturbed joint strategy set $\cX^{(\epsilon)}$, and we execute gradient descent long enough so that there is a strategy $\vx(t)$ such that $\langle \nabla \Phi(\vx(t)), \vx' - \vx(t) \rangle \geq - \epsilon'$ for any $ \vx' \in \cX^{(\epsilon)}$; this can be converted into an $\epsilon$-perfect $O_{\epsilon'}(\sqrt{\epsilon'})$-almost equilibrium per the well-supported notion~\citep{Chen09:Settling}.

\paragraph{Polymatrix games} We now observe that~\Cref{lemma:almost-near} can be significantly improved in polymatrix games. This refinement will allow us to obtain a more direct $\CLS$ membership for computing perfect equilibria in that class of games. We begin with the following simple observation relating perfect equilibria in polymatrix games and \emph{linear complementarity problems (LCPs)}; we refer to~\citet{Hansen18:Computational} for an analogous connection pertaining to the more refined notion of proper equilibria.

\begin{lemma}
    \label{lemma:polymatrix-LCPs}
    Any perfect equilibrium of a polymatrix game can be obtained as a limit point, as $\epsilon \to 0$, of a solution to the perturbed standard-form LCP, $\mathcal{P}^{(\epsilon)}$, defined as
\begin{equation*}
\begin{array}{ll}
\text{find}  & \vec{z}, \vec{w} \in \R^d \\
\text{subject to}& \vec{z}^\top \vec{w} = 0, \\
& \vec{w} = \mat{M}^{(\epsilon)} \vec{z} + \vec{b}^{(\epsilon)}, \\
& \vec{z}, \vec{w} \geq 0.
\end{array}    
\end{equation*}
Furthermore, every entry of $\mat{M}^{(\epsilon)}$ and $\vec{b}^{(\epsilon)}$ has degree at most $1$ as a polynomial in $\epsilon$ and each coefficient can be represented with a polynomial number of bits.
\end{lemma}

We recall that a \emph{basis} $B$ for a standard-form LCP with constraints $\vec{w} = \mat{M} \vec{z} + \vec{b}$ is a set of linearly independent columns of $\mat{M}$ such that the associated (complementary) solution---called \emph{basic solution}---is feasible. In this context, we borrow the following definition of~\citet{Farina17:Extensive}.

\begin{definition}[Negligible positive perturbation; NPP]
    Let $\mathcal{P}^{(\epsilon)}$ be an LCP parameterized on some perturbation parameter $\epsilon$. The value $\epsilon^* > 0$ is a \emph{negligible positive perturbation (NPP)} if any (complementary) basis $B$ for $\mathcal{P}^{(\epsilon^*)}$ remains a basis for $\mathcal{P}^{(\epsilon)}$ for all $0 \leq \epsilon \leq \epsilon^*$.
\end{definition}

\citet{Farina17:Extensive} gave a general result regarding perturbed LCPs per~\Cref{lemma:polymatrix-LCPs}.

\begin{lemma}[\citep{Farina17:Extensive}]
    \label{lemma:NPP}
    Suppose that each entry of $\mat{M}^{(\epsilon)}$ and $\vec{b}^{(\epsilon)}$ has degree $\poly(d)$ and each coefficient can be represented with $\poly(d)$ bits. Then there exists some negligible positive perturbation $\epsilon^* \in \Q_{>0}$ expressed in $\poly(d)$ bits for $\mathcal{P}^{(\epsilon)}$.
\end{lemma}

We thus arrive at the following consequence.

\begin{proposition}
    \label{prop:polymatrix-round}
    For any polymatrix game $\Gamma$, there exists a polynomial and an $\epsilon \leq 2^{p(|\Gamma|)}$ such that any $\epsilon$-perfect equilibrium of $\Gamma$ induces an exact perfect equilibrium of $\Gamma$.
\end{proposition}

This makes the problem of computing an exact perfect equilibrium of a potential, polymatrix game $\Gamma$ directly amenable to gradient descent, establishing an alternative proof of $\CLS$ membership. Indeed, one can run gradient descent (\Cref{prop:gd-folklore}) on the polymatrix potential game $\Gamma^{(\epsilon^*)}$ to identify an $\epsilon^*$-perfect equilibrium through a standard rounding procedure~\citep[Appendix A]{Fearnley24:Complexity}; since $\epsilon^* \leq 2^{p(\Gamma)}$, $\Gamma^{(\epsilon^*)}$ can be equivalently expressed as a polymatrix game (without any perturbation in the strategy sets) where each entry in the payoff matrices has polynomially many bits. The same argument goes through for proper equilibria in polymatrix games.

\paragraph{Symbolic gradient descent} While~\Cref{prop:polymatrix-round} shows that a numerical version of gradient descent, executed on the perturbed game for a sufficient number of iterations, will succeed, we now point out that \emph{symbolic} gradient descent can fail. This can be seen even in the simple example of~\Cref{fig:symbolic}, as we point out in~\Cref{prop:symbolicfails}.

We first clarify what we mean by ``symbolic gradient descent.'' Every player $i \in [n]$ is assumed to initialize at some symbolic strategy $\vx^{(\epsilon)}_i(0) \in \cX^{(\epsilon)}$. Then $\vx^{(\epsilon)}(t+1)$ is obtained from $\vx(t)$ by applying symbolically a projected gradient descent step using $\nabla \Phi(\vx^{(\epsilon)}(t))$. It is easy to see that a symbolic gradient descent step can be performed efficiently on the truncated probability simplex~\citep{WangC13:Projection}. After executing $T$ iterations, we return as output $\vx^{(0)}(T)$, which is equal to $\lim_{\epsilon \to 0^+} \vx^{(\epsilon)}(T)$ by continuity.

\begin{proposition}
    \label{prop:symbolicfails}
     There is a $2 \times 2$ game such that for any $T \in \N$, symbolic gradient descent after $T$ iterations converges to a strategy profile that is not a perfect equilibrium.
\end{proposition}

\begin{proof}
    We analyze the game given in~\Cref{fig:symbolic} when we initialize at $\vx^{(\epsilon)}_1(0) = (1 - \epsilon, \epsilon)$ and $\vx^{(\epsilon)}_2(0) = (1 - \epsilon, \epsilon)$. We will prove the invariance $\vx^{(\epsilon)}_1(t) = (1 - \epsilon \kappa(t), \epsilon \kappa(t))$ and $\vx^{(\epsilon)}_2(t) = (1 - \epsilon \kappa(t), \epsilon \kappa(t))$ for some $\kappa(t)$. Inductively, if it holds for $t$, we have
    \begin{equation*}
        \vx^{(\epsilon)}_1(t+1) = \vx^{(\epsilon)}_2(t+1) =  \Pi_{\Delta(2)} \left(
        \begin{pmatrix}
         1 - \epsilon \kappa(t)\\
         \epsilon \kappa(t)
        \end{pmatrix} + \begin{pmatrix}
         0\\
         \epsilon \kappa(t)
        \end{pmatrix} \right) = \begin{pmatrix}
         1 - \frac{3}{2} \epsilon \kappa(t)\\
         \frac{3}{2} \epsilon \kappa(t)
        \end{pmatrix}
    \end{equation*}
    for all $\epsilon > 0$ small enough, establishing the induction for $\kappa(t) = ( \nicefrac{3}{2} )^{t}$. It thus follows that, no matter the choice of $T$, $\vx_1^{(0)} = (1, 0) = \vx_2^{(0)}$, which is not a perfect equilibrium.
\end{proof}

\begin{figure}[!h]
\centering
\begin{tabular}{c c c}
& \textsf{C1} & \textsf{C2} \\
\textsf{R1} & 0, 0 & 0, 0 \\
\textsf{R2} & 0, 0 & 1, 1 \\
\end{tabular}
\caption{A $2 \times 2$ identical-interest game in normal form where symbolic gradient descent can fail to converge to a perfect equilibrium.}
\label{fig:symbolic}
\end{figure}

\subsection{Doubly exponentially small $\eps$ is necessary}
\label{sec:doublyexpo}

We have now seen that, for polymatrix, potential games, the almost-implies-near framework implies that computing perfect equilibria lies in the complexity class $\CLS$. One may hope that this applies more generally to potential games, as potential games have more structure than generic normal-form games that may hypothetically allow a stronger almost-to-near result, even without the polymatrix assumption. Here, we show otherwise: we will explicitly exhibit normal-form games $\Gamma$ in which $\eps$ is required to be {\em doubly-exponentially small} (in the representation size of the game) before every (even exact) Nash equilibrium of $\Gamma^{(\eps)}$ is close to a (even Nash) equilibrium of $\Gamma$.


\begin{theorem}
    For every positive integer $n$ there exists a normal-form potential game $\Gamma_n$ with $O(n)$ players and two actions per player, or a polytope potential game with three players each with strategy set $[0, 1]^{O(n)}$, such that, for all $\eps \in [1/2^{2^n}, 1/2]$, the perturbed game $\Gamma^{(\eps)}_n$ admits a Nash equilibrium that is distance $1/2$ away in $\ell_\infty$-norm from any Nash equilibrium of $\Gamma_n$.
\end{theorem}
\begin{proof}
We will construct a normal-form game $\Gamma_n$ with $4n+1$ players, as follows. The first $4n$ players are split into four groups of $n$ players each, and the mixed strategies of the four groups will be denoted $\vx, \vx', \vc, \vd \in [0, 1]^n$, where for example $c_i \in [0, 1]$ is the probability that the $i$th player of group $\vc$ plays its first action. The last player's strategy will be denoted $t \in [0, 1]$. The potential function is given by 
    \begin{align}
        \Phi(\vx, \vx', \vc, \vd, t) = \sum_{i=1}^n \qty[\qty(t - c_i) (x_i - x_{i-1} x_{i-1}') + \qty(d_i - \frac12) (x_i - x_i')] - 2x_n - 2x'_n - 2n \cdot t \label{eq:doubleexp}
    \end{align}
where for simplicity of notation we set $x_0 = x_0' := 1/2$. 

The construction of the three-player game is identical: the three players have strategies $\vx_i \in [0, 1]^n, \vx_i' \in [0, 1]^n$, and $(\vc, \vd, t) \in [0, 1]^{2n+1}$, respectively.

For intution, $\Phi$ is constructed to achieve strategies at equilibrium that involve repeated squaring, as follows. If $\vd = \vec 1/2$ and $\vec c = t\vec 1$, then we must have $x_i = x_{i-1} x_{i-1}'$ and $x_i = x_i'$, which implies $x_i = 1/2^{2^i}$. As mentioned in the body, the sole purpose of the player $t$ is to have a value that is guaranteed to be $\eps$, since by construction $t$ always has negative gradient. 

We claim first that, for $\eps \in [1/2^{2^n}, 1/2]$, the perturbed game $\Gamma_n^{(\eps)}$ has the following equilibrium: $$t = \eps, \quad \vd = \frac{\vec 1}2, \quad \vc = \eps \vec 1, \quad x_i = x_i' = \max\qty{\eps, \frac{1}{2^{2^i}}} = \max\{\eps, x_{i-1}^2\}.$$ We check this by checking the gradients for all players:
\begin{itemize}
    \item The gradients for $x_n$, $x_n'$, and and $t$ are always negative, and they are already set to the smallest possible value $\eps$.
    \item The gradient for $c_i$ is nonpositive, because $x_i \ge x_{i-1}^2$ holds for all $i$, and $c_i$ is set to the smallest possible value, $\eps$.
    \item For $i < n$, the gradients for $x_i$, $x_i'$, and $d_i$ are all $0$, because $c_i = t$, $d_i = 1/2$, and $x_i = x_i'$ for all $i$. 
\end{itemize}

Thus, this is an equilibrium. It remains to show that this equilibrium is not close to any Nash equilibrium of $\Gamma_n$. Indeed, we will show that no equilibrium of $\Gamma_n$ can satisfy $0 < d_i < 1$ and $c_i < 1$ for all $i$. Since $0 < d_i < 1$, we must have $x_i = x_i'$ for all $i$. Since the gradients for $x_n$, $x_n'$, and and $t$ are always negative, we must have $x_n = x_n' = t = 0$. But then there must be some $i$ for which $x_i < x_{i-1}^2$. For that value of $i$, the gradient for $c_i$ will be positive, so we must have $c_i = 1$.  
Thus, every Nash equilibrium of $\Gamma_n$ must have $c_i = 1$ or $d_i \in \{0, 1\}$ for some $i$, and therefore be at least $\ell_\infty$ distance $1/2$ from the previously established equilibrium of $\Gamma^{(\eps)}_n$.
\end{proof}

\begin{theorem}
    \label{theorem:3playerdoublyexp}
    For every positive integer $n$ there exists a normal-form potential game $\Gamma_n$ with $3$ players and $O(n)$ actions per player, such that, for all $\eps \in [1/2^{2^n}, 1/4n]$, the perturbed game $\Gamma^{(\eps)}_n$ admits a Nash equilibrium that is distance $\Omega(1/n)$ away in $\ell_\infty$-norm from any Nash equilibrium of $\Gamma_n$.
\end{theorem}
\begin{proof}
    It suffices to only slightly modify the three-player counterexample in the previous proof; for the sake of completeness, we write out the full argument. We add one action for each of the three players, which will be denoted $x_{n+1}, x_{n+1}'$, and $s$, respectively. So, the strategy spaces of the three players are $\vx \in \Delta(n+1)$, $\vx' \in \Delta(n+1)$, and $(\vc, \vd, t, s) \in \Delta(2n+2)$, respectively. The extra actions $x_{n+1}, x_{n+1}'$, and $s$ serve only to ensure the existence of an extra degree of freedom in the strategy spaces, and will be essentially ignored for the remainder of the argument. The potential function is set to 
    \begin{align*}
        \Phi(\vx, \vx', \vc, \vd, t, s) = \sum_{i=1}^n \qty[\qty(t - c_i) (x_i - x_{i-1} x_{i-1}') + \qty(d_i - \frac1{4n}) (x_i - x_i')] - 2x_n - 2x'_n - 2n \cdot t
    \end{align*}
    The only change from \eqref{eq:doubleexp} is that the $d_i - 1/2$ has become a $d_i - 1/2n$, which is only necessary to ensure that P3's strategy is actually a valid strategy in the remainder of the proof. No attempt is made to optimize constant factors.

We claim first that, for $\eps \in [1/2^{2^n}, 1/4n]$, the perturbed game $\Gamma_n^{(\eps)}$ has the following equilibrium: $$t = \eps, \quad \vd = \frac{\vec 1}{4n}, \quad \vc = \eps \vec 1, \quad  x_i = x_i' = \max\qty{\eps, \frac{1}{2^{2^i}}} = \max\{\eps, x_{i-1}^2.\}$$ (To see that this is a valid strategy profile, notice that $\sum_{i=1}^n \max\{\eps, x_{i-1}^2\} \le \eps n + 1/2 \le 1$, so it suffices for $\eps \le 1/2n$.) We check this by checking the gradients for all players:
\begin{itemize}
    \item The gradients for $x_n$, $x_n'$, and and $t$ are always negative, and they are already set to the smallest possible value $\eps$.
    \item The gradient for $c_i$ is nonpositive, because $x_i \ge x_{i-1}^2$ holds for all $i$, and $c_i$ is set to the smallest possible value, $\eps$.
    \item For $i < n$, the gradients for $x_i$, $x_i'$, and $d_i$ are all $0$, because $c_i = t$, $d_i = 1/2$, and $x_i = x_i'$ for all $i$. 
\end{itemize}

Thus, this is an equilibrium. It remains to show that this equilibrium is not close to any Nash equilibrium of $\Gamma_n$. Indeed, we will show that no equilibrium of $\Gamma_n$ can satisfy $1/5n \le d_i \le 1/3n$ and $c_i \le 1/3n$ for all $i$. Suppose there were such an equilibrium. Since the gradients for $x_n$, $x_n'$, and and $t$ are always negative, we must have $x_n = x_n' = t = 0$. Thus, we have $1/3 \le s \le 4/5$, so the gradient for $d_i$ must be zero and the gradient for $c_i$ must be nonpositive for all $i$ (else P3 profitably deviates by moving mass between $s$ and the $d_i$ or $c_i$ whose gradient is nonzero.) But then we must have $x_i = x_i'$ for all $i$. But then there must be some $i$ for which $x_i < x_{i-1}^2$. For that value of $i$, the gradient for $c_i$ will be positive which is a contradiction. 
Thus, every Nash equilibrium of $\Gamma_n$ must have $c_i = 1$ or $d_i \in \{0, 1\}$ for some $i$, and therefore be at least $\ell_\infty$ distance $1/2$ from the previously established equilibrium of $\Gamma^{(\eps)}_n$.
\end{proof}

\section{Exponential path lengths for pure perfect equilibria}
\label{sec:longpaths}

In this section, we investigate how many steps $\eps$-symbolic better-response dynamics might take in comparison to better-response dynamics. Of course, solution points to $\eps$-symbolic better-response dynamics are in particular solution points to better-response dynamics, that is,  pure perfect equilibria are pure Nash equilibria. Since pure Nash equilibria are \PLS{}-hard to compute in concise potential games (see, \emph{e.g.}, \citet{Fabrikant04:Complexity} for congestion games), pure perfect equilibria will inherit that complexity hardness. Therefore, one might intuitively hope to argue---and mistakenly so, as we will see---that $\eps$-symbolic better-response dynamics should take at least as long to converge. We will show instead that \emph{either} dynamics can take exponentially longer than the other. We remark that these results are \emph{unconditional}, \ie, true independent of whether $\PLS{} = \Pol{}$ or not. We will first present the results of this section, and then discuss further background and the proofs.


For the specificity of our results, we recall our assumption that in both dynamics, players update their strategy in a round-robin fashion. On the positive side, we can then show that $\eps$-symbolic better-response dynamics can be exponentially faster than better-response dynamics in terms of improvement steps needed to reach an equilibrium.

\begin{theorem}
\label{thm:short eps br paths}
    There are families of concise potential games in which 
    \begin{enumerate}
        \item $\eps$-symbolic better-response dynamics from any starting action profile takes at most $n$ improvement steps to reach a pure perfect equilibrium, whereas
        \item there exists starting action profiles from which better-response dynamics takes exponentially many improvement steps to reach a pure Nash equilibrium.   
    \end{enumerate}
    This already holds for games in which all players have exactly two actions.
\end{theorem}
From the fact our proof constructs a family of $2$-action games, we also obtain that (1) \Cref{thm:short eps br paths} holds true independent of the pivoting rule a player might deploy for choosing a better response, and that (2) the results extend to proper equilibria as well (formally stated in \Cref{cor:exppaths for proper}).

We prove \Cref{thm:short eps br paths} by reducing from the local search problem {\sc MaxCut/Flip}. It is a \PLS{}-complete problem that is known to possibly admit exponentially long local search paths. Interestingly enough, we will \emph{also} use {\sc MaxCut/Flip}---and a so-called \emph{tight} \PLS{}-reduction---to prove the contrary observation. 

\begin{theorem}
\label{thm:tight plsh}
    There is a \emph{tight} \PLS{}-reduction from {\sc MaxCut/Flip} to finding a pure perfect equilibrium of concise potential games. This already holds for games in which all players have exactly two actions.
\end{theorem}

The proof of~\Cref{thm:tight plsh} also implies the following.

\begin{corollary}\label{cor:shortBRpaths}
    There are families of concise potential games in which 
    \begin{enumerate}
        \item better-response dynamics from any starting action profile takes at most $n$ improvement steps to reach a pure Nash equilibrium, whereas
        \item there exists starting action profiles from which $\eps$-symbolic better-response dynamics takes exponentially many improvement steps to reach a pure perfect equilibrium.   
    \end{enumerate}
\end{corollary}

\begin{corollary}
\label{cor:PSPACEcforepssymb}
    It is \PSPACE{}-complete to decide in concise potential games whether $\epsilon$-symbolic better-response dynamics reaches a pure perfect equilibrium from a given starting action profile within $k \in \N$ number of improvement steps (where $k$ is given in binary).
\end{corollary}

Again, since our proof constructs a family of $2$-action games, the above results hold true independent of the pivoting rule a player might deploy for choosing an $\eps$-symbolic better response. Also, when players only have two actions available, the notions of pure perfect equilibrium and pure proper equilibrium coincide. Therefore, we conclude with the same statements for the proper equilibrium refinement.
\begin{corollary}
    \label{cor:exppaths for proper}
    All results in this section also hold for pure proper equilibria instead of pure perfect equilibria. This includes \Cref{thm:short eps br paths}, \Cref{thm:tight plsh}, \Cref{cor:shortBRpaths}, and \Cref{cor:PSPACEcforepssymb}.
\end{corollary}

\paragraph{Background on {\sc MaxCut} and \PLS{}-tight reductions}

To show these results, we use the local search problem {\sc MaxCut/Flip}. It is a \PLS{}-complete problem known to have exponentially long local search paths in the worst case. Let $G = (V,E)$ be an undirected graph, $w : E \to \N$ positive edge weights, and $V = B \sqcup C$ a vertex partition into two sets. Then, the cut of $B \sqcup C$ is defined as all the edges in between $B$ and $C$:
\begin{align*}
    &E \cap (B,C) := \{ \{u,v\} = e \in E : u \in B \wedge v \in C \, \textnormal{ or } \, u \in B \wedge v \in C \} \, .
\end{align*}
Its weight is defined as $w(B,C) := \sum_{e \in E \cap (B,C)} w( e )$. The FLIP neighbourhood of partition $B \sqcup C$ is the set of partitions $(B',C')$ that can be obtained from $(B,C)$ by just moving one vertex from one part to the other:
\begin{align*}
    &\textnormal{FLIP}(B,C) := \Big\{ (B \cup \{c\} ) \sqcup ( C \setminus \{c\} ) \Big\}_{c \in C} \cup \Big\{ ( B \setminus \{b\} ) \sqcup (C \cup \{b\} ) \Big\}_{b \in B} \, .
\end{align*}

\begin{definition}
    An instance of the search problem {\sc MaxCut/Flip} consists of an undirected graph $G = (V,E)$ with edge weights $w : E \to \N$. A (locally optimal) solution consists of a partition $V = B \sqcup C$ that has maximal cut weight among its \textnormal{FLIP} neighbourhood.
\end{definition}
We are interested in its computational complexity in terms of $|V|$, $|E|$, and a binary encoding of all weight values. The standard local search method with Flip iteratively checks for the current cut $(B,C)$ whether there is a node to flip that would improve upon the current cut weight. A pivoting rule decides which node to flip for improvement if there are multiple ones available. The next lemma shows that local search fully captures the complexity of this problem.

\begin{lemma}[\citet{SchafferY91,Yannakakis2003}]
\label{lem:maxcut PLS-c and implications}
    {\sc MaxCut/Flip} is \PLS{}-complete. Moreover, there are instances and initializations (starting partitions) for which---irrespective of the choice of pivoting rule--- local search takes exponentially many iterations to reach a solution. Deciding whether one can reach a solution from a given initialization within $k \in \N$ number of iterations ($k$ given in binary) is \PSPACE{}-complete.
\end{lemma}

The latter two results in \Cref{lem:maxcut PLS-c and implications} can be obtained with a so-called \emph{tight} \PLS{}-reduction. Since it forms a central pillar to our upcoming results, we shall quickly restate its definitions here (\cf \citealp[Definitions 2-4]{Yannakakis2003}).

\begin{definition}
    A \PLS{}-reduction from a local search problem $\Pi$ to a local search problem $\Pi'$ consists of two polytime computable functions $h$ and $g$ such that
    \begin{enumerate}
        \item $h$ maps instances $x$ of $\Pi$ to instances $h(x)$ of $\Pi'$,
        \item $g$ maps a pair consisting of an instance $x$ of $\Pi$ and a feasible point of $h(x)$ to a feasible point of $x$, and
        \item if $s$ is a (locally optimal) solution to $h(x)$, then $g(x,s)$ is a solution to $x$.
    \end{enumerate}
\end{definition}

\begin{definition}
    Let $\Pi$ be a local search problem, and $x$ be an instance of $\Pi$. Its transition graph $\textnormal{TG}_\Pi(x)$ is then defined as a directed graph with one node for each feasible point to $x$, and an arc $s \to t$ whenever $t$ is in the local neighborhood of $s$ and the cost of $t$ is strictly better than the cost of $s$.
\end{definition}
For example, for an instance $(G,w)$ of the local search problem {\sc MaxCut/Flip}, the partitions are the feasible points, and there is an arc in the transition graph from a partition $B \sqcup C$ to a partition $B' \sqcup C'$, if the latter has a strictly higher cut weight and was obtained from the former via a vertex flip.
\begin{definition}
\label{defn:tight PLS-red}
    Let $(h,g)$ be a \PLS{}-reduction from local search problem $\Pi$ to local search problem $\Pi'$. We say it is \emph{tight} if for any instance $x$ of $\Pi$ we can choose a subset $R$ of feasible points of corresponding instance $h(x)$ of $\Pi'$ such that
    \begin{enumerate}
        \item $R$ contains all (locally optimal) solutions to $h(x)$,
        \item for any feasible point $p$ of $x$ we can construct in polytime a preimage $q$ of $p$ under $g(x, \cdot)$, \ie, $g(x,q) = p$, and
        \item the following holds: If the transition graph $\textnormal{TG}_{\Pi'}(h(x))$ of $h(x)$ contains a path from $q \in R$ to $q' \in R$ such that all of its internal path nodes are outside of the path, then either $g(x,q) = g(x,q')$ already or $\textnormal{TG}_{\Pi}(x)$ contains an arc from $g(x,q)$ to $g(x,q')$.
    \end{enumerate}
\end{definition}

A tight \PLS{}-reduction from $\Pi$ to $\Pi'$ implies that the transition graphs in the corresponding $\Pi'$ instances are at least as long as the transition graphs in their respective $\Pi$ instances. This enables lower bound proofs on the local search method as given in \Cref{lem:maxcut PLS-c and implications} assuming they are known to hold for the problem $\Pi$ one reduced from \citep[Lemma 11]{Yannakakis2003}. With that, we have the necessary tools to show the results of this section.

\subsection{Proof of \Cref{thm:short eps br paths}}
This subsection is devoted to proving \Cref{thm:short eps br paths}. We first review a result known in the literature, as it forms the basis of our proof.
\begin{lemma}
\label{lem:known maxcut to II game}
    There is a \emph{tight} \PLS{}-reduction from {\sc MaxCut/Flip} to finding a Nash equilibrium in identical-interest games. In particular, there exist identical-interest games and action profiles $\va$ in these games from which all better-response paths are of exponential length.
\end{lemma}
\begin{proof}
    Let $(V,E,w)$ be a {\sc MaxCut/Flip} instance, and create an identical-interest game of it as follows: Introduce a player $v \in V$ for each node. Each player has two actions $\{b,c\}$, which represent entering the subset $B$ or $C$ in the partition. Hence, a pure strategy profile $\va$ corresponds to a partition $(B_\va,C_\va)$ of the vertex set. Define every players utility from that profile to be the weight of the partition's cut. Then, the Nash equilibria of that game corresponds to locally optimal partitions of the graph with respect to the FLIP neighborhood structure (``Can any vertex improve by switching to the other side of the partition?''). The tightness of this reduction follows from the fact that the transition graphs are exactly the same. 
\end{proof}
To account for equilibrium refinements, we start from the identical-interest game construction above, and add two players $w, w'$, each with two actions \textsf{d} (default) and \textsf{e} (escape). If at least one of those players takes action \textsf{d}, we compute the payoff as we did before, based on the cut weight of the action profile of the other players. If $w$ and $w'$ play \textsf{e}, on the other hand, we give everyone a very high payoff $M + \psi(\va)$---so high that they would not be able to achieve a comparable payoff through the weight of a cut. (For example, we can choose $M = |V|^2 \cdot \max_e w(e) + 1$.) We define the additional payoff $\psi(\va)$ to be $\big|\big\{v \in V: \va_{v} = b \big\}\big|$. Then, pure Nash equilibria of this game have one of two forms: either $(\vb, \textsf{e}, \textsf{e})$, where $\vb$ is the action profile where the original $|V|$ players are playing $b$; or $(\va, \textsf{d}, \textsf{d})$, where $\va$ represents a local max-cut of the graph. Only $(\vb, \textsf{e}, \textsf{e})$ is also pure perfect equilibria, since for the latter type, $w$ observes \textsf{e} to be better than \textsf{d} whenever $w'$ plays \textsf{e} with some nonzero probability.

Let $w$ and $w'$ be first to update their strategy. Then all starting action profiles have an $\eps$-symbolic better-response paths of length at most $|V|$ to a pure perfect equilibrium: $w$ updates to \textsf{e} due to the first-order incentives from $w'$ switching as well, $w'$ follows suit due to zeroth order incentives, and the remaining players update to $b$. Now consider better-response dynamics for the starting profile $(\va, \textsf{d}, \textsf{d})$. As long as both $w$ and $w'$ play \textsf{d}, switching to action \textsf{e} will never become a better response for either of those players. Hence, any better-response path starting from $(\va, \textsf{d}, \textsf{d})$ will correspond to better-response paths in the original identical-interest game, appended by $w$ and $w'$ playing \textsf{d}. All of these paths will be of exponential length if the identical-interest games and starting profile $\va$ are chosen according to \Cref{lem:known maxcut to II game}.

\subsection{Proof of~\Cref{thm:tight plsh} and its subsequent corollaries}

This subsection is mainly devoted to proving \Cref{thm:tight plsh}. We reduce from {\sc MaxCut/Flip} again, so let $(V,E,w)$ be a {\sc MaxCut/Flip} instance. We create a similar identical-interest game as in the proof of \Cref{thm:short eps br paths}, except with a different twist. For each node $v \in V$, introduce \emph{three} players $v^{(1)}, v^{(2)}$, and $v^{(3)}$, each of them with the usual two actions $\{b,c\}$. Informally, we want the triplet to decide via majority vote whether $v$ should enter the subset $B$ or $C$ in the partition, but if $k$ triplets do not unanimously agree and $k \geq 2$, we will punish them slightly and proportionally to $k$ (and everyone else equally so, since the game shall be of identical interest). For an action profile $\va = (\va_{v}^{(1)},\va_{v}^{(2)},\va_{v}^{(3)})_{v \in V} \in A$ in the game, define its associated cut $(B_\va, C_\va)$ as follows: put vertex $v \in B$ if and only if at least two players of the triplet $\{v^{(1)}, v^{(2)}, v^{(3)}\}$ play $b$ (note that otherwise, at least two must have played $c$). Next, define $\psi: A \to \N$ as the number non-unanimous triplets, that is, $\psi(\va) = \big|\big\{v \in V: \exists i,j \in [3] \text{ with } \va_{v}^{(i)} \neq \va_{v}^{(j)} \big\}\big|$. Lastly, let $\lambda$ be a sufficiently large yet polynomially sized penalty multiplier; for example, $\lambda = 6|V|$ suffices. We then define the utility function as
\begin{align*}
    u(\va) = 
        \begin{cases}
            w(B_\va, C_\va) &\text{if } \psi(\va) \leq 1,\\
            w(B_\va, C_\va) - \psi(\va)/\lambda &\text{otherwise}.
        \end{cases}
\end{align*}

We denote this corresponding game instance with $\Gamma$. For purposes of analyzing its game dynamics, we assume that the player order has triplets appearing back to back. With this, we provide some intuition for the Nash equilibrium structure in the next lemma.

\begin{lemma}
    An action profile $\va$ is a pure Nash equilibrium of this corresponding game $\Gamma$ if and only if it satisfies $\psi(\va) \leq 1$ and the following property: if there is a node $v$ with a non-unanimous triplet, then its majority is playing the (weakly) better partition subset for $v$ among $B$ and $C$. Furthermore, better-response dynamics from any starting action profile in $\Gamma$ takes at most as many improvement steps to reach a pure Nash equilibrium as there are players in $\Gamma$.
\end{lemma}
\begin{proof}
    Indeed, if the two properties are satisfied, then (a) players of unanimous triplets have no incentives to unilaterally deviate (according to better-response improvements) because we are using majority voting, and (b) the same is true for non-unanimous triplets if it already voted for the (weakly) better subset. For the other direction, let us discuss the situations not covered by the two situations. If there are at least two non-unanimous triplets, then the minority player of any non-unanimous triplet has incentives to join the other two player's decision in order to push down the punishment term. If there is only one non-unanimous triplet, and it is voting for a strictly worse subset, then any of the majority players have incentives to switch to the minority player's decision to turn the majority vote. This concludes the other direction of the Nash equilibrium characterization. Last but not least, the statement about better-response dynamics follows from the fact that non-unanimous triplets will come to unanimous agreement first until there is only one non-unanimous triplet left, and that one will turn its majority at most once.
\end{proof} 

Next, we can show that the construction of corresponding $\Gamma$ gives rise to a tight \PLS{}-reduction from {\sc MaxCut/Flip} to finding a pure perfect equilibrium of concise potential games, that is, \Cref{thm:tight plsh}. First, we show it gives rise to a \PLS{}-reduction in the first place. As mentioned earlier, we associate any action profile $\va$ to the partition $(B_\va, C_\va)$ that one obtains by considering the majority vote for each triplet. Then we have the following.
\begin{lemma}
\label{lem:unanimity}
    If $\va^*$ is a pure perfect equilibrium of $\Gamma$ 
    then $\psi(\va^*) = 0$.
\end{lemma}
\begin{proof}
    It cannot be $\psi(\va^*) \geq 2$ because a minority player of any non-unanimous triplet would see zero-order incentives to deviate to reach unanimity. \emph{Zero-order} here means a strict positive utility gain from a unilateral deviation. It cannot be $\psi(\va^*) = 1$ either, because for the non-unanimous triplet $\{v^{(1)}, v^{(2)}, v^{(3)}\}$ in that situation either (a) a majority player sees zero-order incentives to deviate because $v$ is on currently on the worse side of the partition or, otherwise, (b) the minority player sees first-order incentives to reach unanimity. \emph{First-order} here means zero utility gain from a unilateral deviation, and strictly positive utility gain from one owns deviation under the off-chance that a single other player deviated in line with perfect equilibrium reasoning. The reason the minority player sees first-order incentives here is because the minority player (weakly) wants to avoid a majority player deviating, resulting in a non-beneficial swap for $v$ (since we are not in case (a)), and because the minority player strictly wants to avoid the penalty term from when a player of the other triplets deviated and created another non-unanimity. 
\end{proof}
\begin{lemma}
\label{lem:better maxcut from player deviation}
    If $\va^*$ is a pure perfect equilibrium of $\Gamma$ then $(B_{\va^*}, C_{\va^*})$ is a locally optimal solution to the original {\sc MaxCut/Flip} instance $(V,E,w)$.
\end{lemma}
\begin{proof}
    By \Cref{lem:unanimity}, $\psi(\va^*) = 0$. To show that $(B_{\va^*}, C_{\va^*})$ is locally optimal, we show that no vertex $v$ can be flipped and yield a cut $(B', C')$ with an improvement weight, since that would imply that $\va^*$ could have been a pure perfect equilibrium in the first place. So let us study the deviation incentives of any player of the associated triplet, say, player $v^{(1)}$. That player does not see any zero-order deviation from deviating since we are deploying majority voting. On the first-order level, its incentives to deviate come either from a triplet partner possibly deviating or from player of another triplet deviating. If a player of another triplet deviates, the partition---and, therefore, the cut weight---would stay unchanged, but a penalty term of $\lambda$ would incur on the players. If it is a triplet partner instead that deviates together with $v^{(1)}$, the majority for $v$ turns, leading to a vertex flip in the partition but no penalty term. The first-order incentives summarize to 
    \begin{align}
        \label{eq:fo inc PLSh}
        \eps \cdot \Big( 2 \cdot \big( w(B', C') - w(B_\va, C_\va)\big) + 3 (|V|-1) \cdot (-\frac{2}{\lambda}) \Big) \, .
    \end{align}
    Note that the second part is negative and, by our choice of $\lambda$, strictly smaller in magnitude than $1$. Hence, the first-order incentives will be positive if and only if $w(B', C') > w(B_\va, C_\va)$, \ie{}, if and only if the FLIP neighbor $(B', C')$ is a local improvement over $(B_{\va^*}, C_{\va^*})$ in the MaxCut instance. This concludes the proof that if $\va^*$ is a pure perfect equilibrium, then $(B_{\va^*}, C_{\va^*})$ must be locally optimal.
\end{proof}

Next, we can show that this \PLS{}-reduction is tight. Define the subset $R$ in \Cref{defn:tight PLS-red} as all unanimous strategy profiles in $\Gamma$, that is, $R := \{\va : \psi(\va) = 0 \}$. By \Cref{lem:unanimity}, this set contains all locally optimal solutions of $\Gamma$. Next, we show the third condition of a tight \PLS{}-reduction.
\begin{lemma}
\label{lem:transition graph of triplet game}
    If there is an $\eps$-symbolic better-response path $P$ from a strategy profile $\va \in R$ to another $\va' \in R$ such that all profiles along that path are not in $R$, then there exists a vertex $v$ such that one can obtain $(B_\va', C_\va')$ from $(B_\va, C_\va)$ by flipping $v$.
\end{lemma}
\begin{proof}
    Let $v$ be the vertex of the player that deviated at the first edge of $P$. Refer to the player as $v^{(i)}$. By the proof of \Cref{lem:better maxcut from player deviation}, we obtain that the player deviation must have incurred because the cut from flipping $v$ is a strict cut weight improvement over $(B_\va, C_\va)$. So let us consider the strategy profile $\va''$ in which $v^{(i)}$ has changed their strategy. The associated partition remains unchanged, and we have $\psi(\va'') = 1$ now. At this profile, the players of triplets not associated to $v$ will observe negative zero-order deviation incentives because deviating yields no partition change and yields a penalty term by introducing a second non-unanimous triplet. The other players of the triplet associated to $v$, on the other hand, observe a positive zero-order deviation incentive to join $v^{(i)}$ in order to turn the majority vote. Hence, the second edge of $P$ must be one of these two players switching their strategy. At that new profile, still no player of another triplet wants to deviate by analogous reasoning. The other two players of the triplet that have switched already, do not want to switch back since that yields a worse cut weight. Only the third triplet player sees (first-order) deviation incentives to switch, for the reason to avoid cut weight losses or penalty terms in the off-chance where another player deviates. Hence, the third edge of $P$ is the third player of the triplet joining the other two to obtain unanimity for $v$ again, landing at a strategy profile in $R$. By assumption on the path $P$, we obtain that that strategy profile must be $\va'$. And indeed, $(B_\va', C_\va')$ is obtained from $(B_\va, C_\va)$ solely by flipping $v$.
\end{proof}

This concludes the proof on \Cref{thm:tight plsh}. Informally, \Cref{lem:transition graph of triplet game} shows that $\eps$-symbolic better-response dynamics from an action profile $\va$ with unanimous triplets corresponds bijectively to, and is three times as long as, local search in the {\sc MaxCut/Flip} instance starting from the partition associated to $\va$. This implies that finding a pure perfect equilibrium inherits from {\sc MaxCut/Flip} that improvement paths ($\eps$-symbolic better-response dynamics) can have exponential length and that it is \PSPACE{}-complete to decide whether a solution can be reached within a fixed number of steps, \ie, \Cref{cor:shortBRpaths} and \Cref{cor:PSPACEcforepssymb}. 

Finally, we highlight that we exclusively worked with two-action games in this reduction as well as the previous one for \Cref{thm:short eps br paths}, yielding \Cref{cor:exppaths for proper}.

\section{Positive results for structured games}
\label{sec:positive}

In the previous section, we showed that computing pure perfect equilibria using better-response dynamics can take exponentially many steps. Therefore, one might naturally wonder whether there are structured classes of games, in which pure perfect equilibria can be computed in polynomial time using better-response dynamics. In this section, we show such positive result for two well-studied classes of games, namely symmetric matroid congestion games and symmetric network congestion games.

\subsection{Symmetric matroid congestion games}\label{sec:matroid}
Matroid congestion games are a broad class of congestion games, where the strategy space of each player consists of the bases of a matroid over the set of resources. Before we state the formal definition of such games, we recall some basic background about matroids. 

    \begin{definition}[Matroid]
        \label{def:matroid}
        A tuple $M = (\cR, \cI)$ is a matroid, if $\cR$ is a finite set of resources and $\cI$ is a nonempty family of subsets of $\cR$ with the property that 
        \begin{enumerate}
            \item if $I \in \cI$ and $J \subseteq I$ then $J \in \cI$, and 
            \item if $I, J \in \cI$ and $|J| < |I|$ then there exists $i \in I \setminus J$ with $J \cup \{i\} \in \cI$.
        \end{enumerate}
    \end{definition}

    Given a matroid $M = (\cR, \cI)$, we say that a subset of resources $I \subseteq \cR$ is independent if $I \in \cI$. Otherwise we call it dependent. A maximal independent subset of a matroid $M$ is called a basis of $M$ and its size is the rank of $M$, denoted as $\text{rk}(M)$.

    \begin{definition}[Matroid congestion game]
        \label{def:matroidcongestion}
        A matroid congestion game is a congestion game, where the strategy space of any player $i$ is the set of bases $B_i$ of a matroid $M_i = (\cR_i \subseteq \cR, \cI_i)$.
    \end{definition}

For a common set of resources $\cR$, a matroid $M = (\cR, \cI)$, and a set of bases $B$ of $M$, we denote a symmetric matroid congestion game as $\Gamma = (\cR, B, (d_r)_{r \in \cR})$, where the strategy space of each player is $B$.

Our approach to show fast convergence involves defining an  $\epsilon$-perturbed game and showing that best-response dynamics will converge to a Nash equilibrium of that game in polynomial time. To show this, we will allow players to randomize over their strategy space. Specifically, we let $\cX = \Delta(B)$ denote the set of mixed strategies over the common strategy space, $B$. For a mixed strategy profile $\vx = (\vx_1, \dots, \vx_n) \in \cX^n$ (and by slight abuse of notation) we define the expected cost $c_i$ of player $i \in [n]$ as
\begin{align}
    c_i(\vx_1, \dots, \vx_n) 
    & = \sum_{S_i \in B} \vx_i(S_i) \cdot \E_{S_{-i}} \left[\sum_{r \in S_i} d_r(n_r(S_i, S_{-i})) \right] \nonumber \\
    & = \sum_{S_i \in B} \vx_i(S_i) \cdot \sum_{r \in S_i} \E_{S_{-i}} \left[ d_r(1 + |\{j \neq i: r \in S_j\}| \right] \nonumber \\
    & = \sum_{S_i \in B} \vx_i(S_i) \cdot \sum_{r \in S_i}  \sum_{k=0}^{n-1}  d_r(1+k) \cdot \Pr_{S_{-i}} \left[ \sum_{j \neq i} \mathbbm{1}\{r \in S_j \} = k \right]. \label{eq:expectedcost}
\end{align}
More generally, we let $P_r(k)$ denote the probability that $\sum_{i \in [n]} \mathbbm{1}\{i \; \text{uses} \; r \} = k$. It is easy to show that in general $P_r(k)$ can be computed in polynomial time \cite{Papadimitriou05:Computing}. Specifically, given a mixed strategy profile, $\vx$, we can first compute the probability, $w_{r}(i)$, that player $i$ will use resource $r$. We have $w_{r}(i) = \sum_{S_i \in B} \vx_i(S_i) \mathbbm{1}\{r \in S_i\}$. Given this, we can compute $P_r(k)$ via simple dynamic programming. We first define $s_r(\ell, k)$ to be the probability that the total number of players, out of the first $\ell$ players, using this resource is $k$. Knowing that $s_r(0, 0) = 1$ and $s_r(0, 1) = 0$, we can define the recurrence relation as 
\begin{align}
   s_r(\ell, k) = s_r(\ell-1, k)(1-w_r(\ell)) + s_r(\ell-1, k-1)w_r(\ell). \label{eq:recurrence}
\end{align}
After polynomially many iterations we get the value of $s_r(n, k)$, which is equal to $P_r(k)$.

We are now ready to define our perturbed game. Given a symmetric, matroid congestion game, $\Gamma = (\cR, B, (d_r)_{r \in \cR})$, we define the $\epsilon$-perturbed game, $\Gamma^{(\epsilon)} = (\cR, \cX^{(\epsilon)}, (D_r)_{r \in \cR}))$, such that each player has strategy space $\cX^{(\epsilon)} = \{\vx \in \cX \;\; \text{such that} \;\; (\vx(S) = \epsilon \lor \vx(S) = 1-(|B|-1)\epsilon) \; \; \forall S \in B\}$ i.e., each player is restricted to playing $\eps$-pure strategies. With $D_r: (\cX^{(\epsilon)})^n \to \R$, we denote the expected delay of a resource $r$, 
defined as $D_r(\vx) = \E_S \left[d_r(|\{j \in [n]: r \in S_j\}|) \right] = \sum_{k=1}^{n} d_r(k) P_r(k)$. 

\begin{remark}
    We note that it is \emph{a priori} unclear that the perturbed game $\Gamma^{(\eps)}$ is a matroid congestion game, which is a property that we will need to prove our main result. We circumvent this issue by defining $\tilde{\Gamma}^{(\eps)}$, a game that is equivalent to $\Gamma^{(\eps)}$ (when players are only allowed to play $\eps$-pure strategies) and is also a matroid congestion game. Specifically, we let $\tilde{\Gamma}^{(\eps)} = (\cR, B, \tilde{d}_r^{(\eps)})$, where $\tilde{d}_r^{(\eps)}: \N \to \N[\eps]$, defined as 
    \begin{align}
        \tilde{d}_r^{(\eps)}(k) =\sum_{i=1}^{n} \tilde{P}_{r, k}(i) d_r(i), \label{eq:perturbeddelays}
    \end{align}
    is the expected delay of resource $r$, when exactly $k$ players are using $r$ in their $\eps$-pure strategy.
    
    To define $\tilde{P}_{r,k}$, we first introduce some additional notation. We let $B_r = \sum_{S \in B}\mathbbm{1}\{r \in S\}$ denote the number of occurrences of resource $r$ in the set of bases $B$. If a player chooses to play a subset of resources that contains $r$ with probability $1-|B|\eps$ then the probability with which this player is using resource $r$ is
    \begin{align*}
        h^{(\eps)}_r &= 1-(|B|-1)\eps + \left( B_r -1 \right) \eps \\
        & = 1- \left(B_r-|B| \right) \eps.
    \end{align*}
    Otherwise, the player is playing resource $r$ with probability $f_r^{(\eps)} =  B_r  \eps$.
    We define
    \begin{align}
        \tilde{P}_{r, k}(i) = \sum_{j= \max\{0, i-(n-k)\}}^{\min\{i, k\}}  \binom{k}{j}(h^{(\eps)}_r)^j (1-h^{(\eps)}_r)^{k-j} \cdot \binom{n-k}{i-j}(f_r^{(\eps)})^{i-j}(1-f_r^{(\eps)})^{n-k-(i-j)} \label{eq:tildeprob}
    \end{align}
    to be the probability that the sum of players playing resource $r$ is $i$, when exactly $k$ players are using $r$ in their $\eps$-pure strategy. It is not hard to verify that $\tilde{\Gamma}^{(\eps)}$ is indeed a symmetric matroid congestion game with potential $\tilde{\Phi}^{(\eps)}(S) = \sum_{r \in \cR}\sum_{k=1}^{n_r(S)} \tilde{d}_r^{(\eps)}(k)$.
\end{remark}
We are now ready to show that $\epsilon$-symbolic best-response dynamics in $\tilde{\Gamma}^{(\epsilon)}$ converge to a Nash equilibrium in polynomial time.

    \begin{theorem}
        Let $\Gamma$ be a symmetric matroid congestion game. Then after at most $n^2|\cR|\text{rk}(M)$ $\epsilon$-symbolic best response iterations in the perturbed game, $\tilde{\Gamma}^{(\epsilon)}$, a perfect equilibrium of $\Gamma$ is reached.  Moreover, each iteration can be implemented in time that is polynomial in the size of the input, assuming that we are given 1) a membership oracle (\ie, an oracle that given a subset $I \in 2^\cR$ decides whether $I \in \cI$) and 2) for each resource $r \in \cR$ the number $B_r = \sum_{S \in B}\mathbbm{1}\{r \in S\}$ of bases that contain $r$.
    \end{theorem}

\begin{proof}
    We can assume without loss of generality that the best response of every player is a pure strategy in $\tilde{\Gamma}^{(\eps)}$, i.e., for a best-responding player $i$ there exists an $S_j \in B$ such that $\vx_j(S_j) = 1$ and $\vx_j(S_j') = 0$ for all $S_{j'} \in B \setminus S_j$.  By \eqref{eq:tildeprob}, we know that $\tilde{P}_r(k)$ is a polynomial in $\eps$ of degree at most $n$. Therefore, for concise games,  given $B_r$ for each $r \in \cR$, the cost of each player for a given strategy profile can be determined in polynomial time. Moreover, it is known that computing best responses can be done efficiently using a greedy algorithm \cite{edmonds1971matroids}.

    \sloppy For the proof of convergence, we will first show that $\tilde{\Phi}^{(\eps)}$ decreases (with respect to the ordering $\lesseps$) with every iteration. To see this, we can first rewrite the potential as $\tilde{\Phi}^{(\eps)}(S) = \sum_{i=1}^{n}\sum_{r \in S_i} \tilde{d}_{r}^{(\eps)}(n_{r, \leq i}(S))$, where $n_{r, \leq i}(S)$ is the number of players $j \leq i$ that use resource $r$.
    
    Now for a given strategy profile $S$ suppose that a player $i \in [n]$ has a profitable deviation $S'_i$. Without loss of generality we can take this player to be player $n$. Therefore, we get that
    \begin{align*}
        \tilde{\Phi}^{(\eps)}(S') - \tilde{\Phi}^{(\eps)}(S) &= \sum_{r \in S'_n} \tilde{d}_{r}^{(\eps)}(n_{r,\leq i}(S')) - \sum_{r \in S_n} \tilde{d}_{r}^{(\eps)}(n_{r, \leq i}(S)) \\ 
        & = \sum_{r \in S'_n} \tilde{d}_{r}^{(\eps)}(n_r(S')) - \sum_{r \in S_n} \tilde{d}_{r}^{(\eps)}(n_r(S)) \\
        & = c_n(S')-c_n(S).
    \end{align*}
    Also, from \citet{Ackermann08:Impact}, we know that in a matroid congestion game, the potential can only take $n^2|\cR|\text{rk}(M)$ many values, and therefore best-response dynamics in $\tilde{\Gamma}^{(\eps)}$ converge to a Nash equilibrium, $\vx$. Hence, the corresponding $\eps$-pure strategies, $\vx^{(\eps)}$ also form a Nash equilibrium in $\Gamma^{(\eps)}$. Lastly, we know that $\{\vx^{(\eps)}\}_{\eps \to 0+}$ defines a sequence of $\eps$-perfect equilibria of $\Gamma$, the limit point of which is a perfect equilibrium. This concludes our proof.
\end{proof}

\subsection{Symmetric network congestion games}\label{sec:network}

A network congestion game is a type of congestion game, defined on a graph, where the sets of resources available to each player correspond to paths in the graph. Concretely, we consider a graph $G = (V, E)$, two nodes $a_i, b_i \in V$ for each player $i \in [n]$, and a nondecreasing delay function $d_e: \N \to \N$ for each $e \in E$. A strategy, $S_i$ of player $i$ corresponds to a subset of edges that forms a valid path from $a_i$ to $b_i$. We denote the set of such paths $\cP_{a_i, b_i}$. As in general congestion games, the potential function associated with this class of games is $\Phi(S) = \sum_{e \in E}\sum_{i=1}^{n_e(S)}d_e(i)$.

In a symmetric network congestion game, every agent shares the same strategy space, which corresponds to all paths, between the same pair of vertices, $a, b \in V$.

As before, we will show that a perfect equilibrium can be found in polynomial time, by defining a perturbed game, and computing a Nash equilibrium of this game, this time by reducing the problem to a symbolic min-cost flow problem.

\begin{theorem}
 Let $\Gamma$ be a symmetric network congestion game. Then there exists an algorithm that outputs a perfect equilibrium of $\Gamma$ in time that is polynomial in the size of the input, assuming that we are given $B_e = \sum_{S \in \cP_{a, b}}\mathbbm{1}\{e \in S\}$ for each edge $e$. 
\end{theorem}
\begin{proof}
    First, given a symmetric network congestion game, $\Gamma = (G = (V, E), \cP_{a, b}, \{d_e\}_{e \in E})$, where $\cP_{a, b}$ is the set of paths between vertices $a$ and $b$, we define the $\eps$-perturbed game $\Gamma^{(\eps)} = (G = (V, E), \cX^{(\eps)}_{a, b}, \{d_e\}_{e \in E})$, where $\cX^{(\eps)}_{a, b} = \{\vx \in \cX \;\; \text{such that} \;\; \vx(P) \geq \eps \;\; \forall P \in \cP_{a, b} \}$, \ie each player is using each path with probability at least $\eps$. We can define an equivalent game to $\Gamma^{(\eps)}$ by applying the perturbations directly to the delay functions. Specifically, we consider the game $\tilde{\Gamma}^{(\eps)} = (G = (V, E), \cP_{a, b}, \{\tilde{d}_e^{(\eps)}\}_{e \in E})$, (where $\tilde{d}_e^{(\eps)}$ are defined as in~\eqref{eq:perturbeddelays}), which is now a symmetric network congestion game.
    
    The next step is to reduce the problem of finding an Nash equilibrium in $\tilde{\Gamma}^{(\eps)}$ to a min-cost flow problem. The crux of this proof follows that of \citet{Fabrikant04:Complexity}, who show this reduction in the non-symbolic case. The idea is that given $\tilde{\Gamma}^{(\eps)}$ with graph $G = (V, E)$, we construct a graph, where each edge $e \in E$ is replaced with $n$ parallel edges, each with capacity $1$ and delays  $\tilde{d}_e^{(\eps)}(1), \dots, \tilde{d}_e^{(\eps)}(n)$. We note that given $B_e$ for each edge $e$, we can compute the total cost of each $e$ in polynomial time. We can now write the symbolic min-cost flow problem as a symbolic linear program, as follows
\begin{align}
    \text{minimize} \quad & \sum_{(u, v) \in E} f_{(u, v)} \cdot \tilde{d}_{(u, v)}^{(\eps)} \\
    \text{subject to} \quad 
        & f_{(u, v)} \leq 1 && \tag{capacity constraints} \\
        & f_{(u, v)} = - f_{(v, u)} && \tag{skew symmetry} \\
        & \sum_{v \in V} f_{(u, v)} = 0 \;\;  \forall u \neq a, b && \tag{flow conservation} \\
        & \sum_{v \in V} f_{(a, v)} = n \;\;  \text{and} \;\; \sum_{v \in V} f_{(v, b)} = n \tag{required flow constraints}
\end{align}
    where with $f_{(u, v)}$ we denote the flow of edge $(u, v)$. In terms of computational complexity, it is known that symbolic linear programs can be solved efficiently \cite{Farina18:Practical}.
    
    Since the solution to this linear program also minimizes the potential function $\tilde{\Phi}^{(\eps)}$, such a solution maps to a pure strategy profile, $\vx^{(\eps)}$, in $\tilde{\Gamma}^{(\eps)}$, where no player can unilaterally change their path to decrease their cost. Hence, $\{\vx^{(\eps)}\}_{\eps \to 0+}$ defines a sequence of $\eps$-perfect equilibria of $\Gamma$, the limit point of which is a perfect equilibrium.
\end{proof}

\subsection{Strongly polynomial algorithms and perturbed optimization}
\label{sec:strongly-poly}

A recurring question within the equilibrium refinement literature is this: if an algorithm can be executed on numerical inputs, can it also be run $\epsilon$-symbolically for an entire family of inputs when $\epsilon > 0$ is small enough? We have already encountered basic versions of this problem. The simplest example is implementing symbolic perfect best-response dynamics, which, on top of polynomial interpolation, rests on computing a maximum of a vector; this can be accomplished by a series of comparisons. A more interesting example pertains to symmetric network congestion games, which was covered a moment ago. There, the implementation rests on performing a series of additions and comparisons, which are again directly amenable to symbolic inputs---when $\epsilon$ is small enough. This begs the question: is there a more general class of algorithms for which one can directly guarantee symbolic implementation? We provide one such answer here by making a connection with \emph{strongly polynomial-time} algorithms.

A strongly polynomial-time algorithm is defined in the arithmetic model of computation. Here, the basic arithmetic operations---addition, subtraction, multiplication, division, and comparison---take one unit of time to execute, no matter the sizes of the operands. An algorithm is said to run in strongly polynomial time~\citep{Grotschel93:Geometric} if i) the number of operations in the arithmetic model of computation is bounded by a polynomial in the number of rationals given as part of the input, and ii) the space used by the algorithm is bounded by a polynomial in the size of the input.

We think of a strongly polynomial-time algorithm $\cA$ as a circuit with a polynomial number of gates $|\cG|$. The degree of a rational function is the sum of the degrees of its numerator and denominator. 

\strongpoly*

\begin{proof}
    Suppose that the output of a gate is given by a rational function $p(\epsilon)/q(\epsilon)$ when $\epsilon > 0$ is small enough so that the outputs consistently lie within a single piece. We let $p(\epsilon) = \alpha_0 + \alpha_1 \epsilon + \alpha_2 \epsilon^2 + \dots + \alpha_r \epsilon^r$ and $q(\epsilon) = \beta_0 + \beta_1 \epsilon + \beta_2 \epsilon^2 + \dots + \beta_{d_\cA - r} \epsilon^{d_\cA - r}$ for some $r \leq d_{\cA}$; we can always take $\beta_j = 1$ for some $j$ by scaling both the numerator and the denominator without affecting the underlying rational function. Consider now a series of distinct points $\epsilon_1, \dots, \epsilon_{d_\cA + 1}$, each of which can be described with $\poly(d_{\cA}, |\cG|)$, such that $q(\epsilon_j) \neq 0$ and $y_j = p(\epsilon_j)/q(\epsilon_j)$.\footnote{Whether $\epsilon_j$ is small enough for the piece corresponding to $p(\epsilon)/q(\epsilon)$ to be activated is moot for this argument.} The coefficients of $p(\epsilon)$ and $q(\epsilon)$ can be determined through any solution to the linear system $p(\epsilon_j) - q(\epsilon_j) y_j = 0$ for $j = 1, \dots, d_{\cA} + 1$. Any two rational functions $p/q$ and $p'/q'$ obtained through a solution to this linear system must be equal, excluding the roots of $q$ and $q'$. Indeed, the degree-$d_\cA$ polynomial $p q' - p' q$ must have $d_{\cA} + 1$ roots, which can only happen if it is the zero polynomial. In particular, each coefficient can be expressed with $\poly(d_\cA, |\cG|)$ bits since it is a solution to a linear system with $\poly(d_\cA, |\cG|)$ bit complexity---each $\epsilon_j$ and $y_j$ can be represented with $\poly(d_\cA, |\cG|)$ bits. This implies that each rational function $p(\epsilon)/q(\epsilon)$ has coefficients whose bit complexity is $\poly(d_\cA, |\cG|)$. As a result, we conclude that, by taking $\epsilon^* = 2^{p(d_\cA, |\cG|)}$ sufficiently small, all outputs in the comparison gates will lie in the same piece for any $\epsilon \leq \epsilon^*$. The output of the circuit as a rational function can be thus determined by selecting a sequence of distinct points $\epsilon_1, \dots, \epsilon_{d_\cA + 1} \leq \epsilon^*$ and applying fractional interpolation based on the (numerical) output of the circuit in the corresponding inputs.
\end{proof}
\section{Computing proper equilibria}\label{sec:computingproper}

We are concerned in this section with the complexity of computing proper equilibria in concise (normal-form) potential games.

For ease of notation, in what follows, we will drop the normalizing factor $\frac{1-\epsilon}{1-\epsilon^{m_i}}$ for each player $i$; since normalizing amounts to a uniform scaling of the utilities, this will not affect our algorithm, which only requires a relative comparison of actions' payoffs.

\begin{theorem}\label{thm:br-converges-to-proper}
    \Cref{alg:delta-br} returns a proper equilibrium of $\Gamma$ in finite time. Further, for concise games, \Cref{line:while,line:br} can be implemented in $\poly(n,m)$ time.
\end{theorem}
\begin{proof}
    We begin by analyzing the per-iteration time complexity. First, we will argue that for any current joint strategy $\vx^{(\epsilon)} \in \cX^{(\epsilon)}$ of the perturbed game, an $\epsilon$-pure strategy best-response by any player, yields at least as much utility for that player as any mixed-strategy. 
    To show this, for any given player $i$ let $\veps_{i}^{(\pi)} =  (\epsilon^{\pi(0)}, \epsilon^{\pi(1)}, \ldots, \epsilon^{\pi(m_i-1)})$ be player $i$'s $\eps$-pure strategy corresponding to permutation $\pi$. Since any $\eps$-perturbed strategy must be a convex combination of $\{\veps_{i}^{\pi}: \pi \in S_{m_i-1}\}$, it follows that the maximum utility that player $i$ can guarantee when best responding to joint strategy $\vx^{(\eps)}$ is $\max_{\pi \in S_{m_i-1}} u_i(\veps_{i}^{(\pi)}, \vx_{-i}^{(\eps)})$. 
    Therefore, we know that throughout \Cref{alg:delta-br}, the probability mass that any player $i \in [n]$ assigns to any given $a_i \in \cA_i$ is of the form $\vx_i^{(\eps)}(a_i) = \eps^k$ for some $k \in \{0, \dots, m_i-1\}$. So for any such joint strategy, $\vx^{(\eps)} \in \cX^{(\eps)}$, and any two permutations, $\pi, \pi' \in S_{m_1-1}$, $p_i(\eps) := u_i(\veps_i^{(\pi)}, \vx_{-i}^{(\eps)})$ and $q_i(\eps) := u_i(\veps_i^{(\pi')}, \vx_{-i}^{(\eps)})$ are polynomials of degree at most $2m$. For concise games $p_i(\eps), q_i(\eps)$ can be evaluated in polynomial time using \Cref{lem:polynomial circuit}. To ascertain whether $p(\epsilon) \lesseps q(\epsilon)$, we can compare their coefficients in lexicographic order. This concludes the analysis of \Cref{line:while}.
    
    Now we show that \Cref{line:br} can be implemented efficiently too. For player $i$, we consider the $m_i$-dimensional utility vector $\vu_i(\vx_{-i}^{(\eps)}) = (u_i(a_i, \vx_{-i}^{(\eps)}))_{a_i \in \cA_i}$. Also, we consider the permutation $\pi_i \in S_{m_i-1}$ that corresponds to the sorted utility values of  $\vu_i(\vx_{-i}^{(\epsilon)})$, that is, $u_i(a_i, \vx_{-i}^{(\eps)}) \lesseps u_i(a_i', \vx_{-i}^{(\eps)})$ $\iff \pi_i(a_i) > \pi_i(a_i')$. Then the best-response of player $i$ in the perturbed strategy space must be the one-hot vector that assigns with probability mass $1$ to $\veps^{(\pi)}$. Indeed, if $u_i(a_i, \vx_{-i}^{(\eps)}) \lesseps u_i(a_i', \vx_{-i}^{(\eps)})$, then $\pi(a_i) > \pi(a_i')$ for any valid best-response $\eps^{(\pi)}$. So, \Cref{line:br} in \Cref{alg:delta-br} can be implemented in polynomial time by sorting the elements of $\vu_i(\vx_{-i}^{(\eps)})$.
    
    We recall that from \citet{Kohlberg86:strategic} we know that a Nash equilibrium of an $\eps$-perturbed game $\Gamma_\eps$ is an $\eps$-proper equilibrium of the original game $\Gamma$. Therefore, the proof of convergence follows a similar argument to the corresponding proof in \Cref{thm:br-converges-to-perfect}.
\end{proof}

\begin{theorem}\label{thm:proper-pls}
    Finding a proper equilibrium of a concise potential game is in \PLS.
\end{theorem}
\begin{proof}
    The proof follows similarly to \Cref{thm:perfect-pls}.
\end{proof}

\section{Price of anarchy of perfect and proper equilibria}

Perfect equilibria---and refinements thereof---constitute a subset of Nash equilibria; from a worst-case perspective, this means that they can only improve the solution quality. One can quantify this through the \emph{price of anarchy} framework~\citep{Koutsoupias99:Worst}. Namely, for a game $\Gamma$ with nonnegative utilities and a solution concept $\solcon$, we define
\begin{equation}
    \label{eq:poa}
    \mathsf{PoA}_{\solcon} \defeq \frac{\max_{\vec{a} \in \cA} \SW(\vec{a}) }{ \inf_{ \vec{x} \in \solcon(\Gamma)} \SW(\vx)}.
\end{equation}
Above, $\SW(\vec{a}) = \sum_{i=1}^n u_i(\vec{a}) $ denotes the (utilitarian) social welfare and $\SW(\vx) = \E_{\vec{a} \sim \vx} \SW(\vec{a})$. It is further assumed that the set of outcomes of $\Gamma$ per the solution concept $\solcon$, denoted by $\solcon(\Gamma)$, is such that $\solcon(\Gamma) \neq \emptyset$. We only treat the non-trivial case in which not all utilities are zero; as a result, if $\inf_{ \vec{x} \in \solcon(\Gamma)} \SW(\vx) = 0$, then $\mathsf{PoA}_{\solcon} = \infty$. In what follows, $\mathsf{NashEq}(\Gamma)$, $\mathsf{PerfEq}(\Gamma)$, and $\mathsf{PropEq}(\Gamma)$ denote the set of Nash, perfect, and proper equilibria of $\Gamma$, respectively.

With this vocabulary at hand, we first point out that the simple example of~\Cref{fig:nash-vs-perfect} already establishes a stark separation between Nash and perfect equilibria:

\begin{proposition}
    There exists a $2 \times 2$ identical-interest game such that $\mathsf{PoA}_{\mathsf{NashEq}} = \infty$, whereas $\mathsf{PoA}_{\mathsf{PerfEq}} = 1$.
\end{proposition}

One criticism of the example of~\Cref{fig:nash-vs-perfect} is that the worst-case Nash equilibrium is brittle: gradient descent initialized at random would converge to the optimal equilibrium almost surely. But, as we pointed out in our introduction, an analogous dichotomy persists even if one initializes gradient descent at random (\Cref{fig:quiver_perfect}), which can be formalized through the \emph{average} price of anarchy notion introduced by~\citet{Sakos24:Beating}.

\paragraph{Lower bounds for polynomial congestion games} While there has been a tremendous amount of interest in algorithmic game theory toward characterizing the price of anarchy with respect to Nash equilibria, less is known about equilibrium refinements from that perspective, with some notable exceptions highlighted in~\Cref{sec:related}. 

Here, we make use of the elegant framework of~\citet{Roughgarden14:Barriers} to establish lower bounds for both perfect and proper equilibria. As a proof of concept, we focus on congestion games, perhaps the most well-studied class of games from the perspective of price of anarchy; for consistency with prior work, in what follows, we switch to cost minimization as opposed to welfare maximization. The basic idea of~\citet{Roughgarden14:Barriers} is that one can harness hardness results for approximating the optimal social cost to obtain lower bounds on the price of anarchy of Nash equilibria. For our purposes, we leverage our previous results on the complexity of perfect and proper equilibria to establish similar lower bounds. On the hardness of approximation front, we rely on the tight lower bounds of~\citet{Paccagnan24:Congestion}, stated in~\Cref{theorem:Gairing}.

We first recall that in a congestion game we are given a finite set of resources $\cR$. Players are to select a subset of resources that they intent to use. We denote by $\cA_i \subseteq 2^\cR$ the action set of player $i \in [n]$. The cost of using a resource $r \in \cR$, $\ell_r : \mathbb{N} \to \R_{> 0}$, depends solely on the number of other players simultaneously using that resource. We denote by $n_r(S)$ the number of players using resource $r$ under the joint action profile $S \in \cA_1 \times \dots \times \cA_n$. The social cost is defined as
\begin{equation*}
    \mathsf{SC}(S) \defeq \sum_{i=1}^n \sum_{r \in \cR} \ell_r(n_r(S)).
\end{equation*}

A canonical example of a congestion game is a \emph{network congestion game}, wherein resources correspond to edges of a graph and players are to select a path. We will restrict our attention to congestion games such that the size of each action set $|\cA_i|$ is polynomial in the representation of the congestion game, so that a best response can be directly computed in polynomial time for both perfect and proper equilibria.

\begin{theorem}[\citep{Paccagnan24:Congestion}]
    \label{theorem:Gairing}
    Let $\ell : \mathbb{N} \to \mathbb{R}_{> 0}$ be a nondecreasing and semi-convex function, and
    \begin{equation*}
        \rho_\ell \defeq \sup_{\nu \in \mathbb{N}} \frac{ \E_{v \sim \mathsf{Poi}(\nu)} [v \ell(v)]  }{ \nu \ell(\nu) }.
    \end{equation*}
    For any $\epsilon > 0$, it is \NP-hard to distinguish between $\min_{\vec{a} \in \cA} \mathsf{SC}(\vec{a}) \leq p(|\cI|)$ and $\min_{\vec{a} \in \cA} \mathsf{SC}(\vec{a}) \geq (\rho_\ell - \epsilon) p(|\cI|)$, where $p(|\cI|)$ is some polynomial in the size of the input $|\cI|$.
\end{theorem}

Above, we remind that a function $\ell : \mathbb{N} \to \mathbb{R}_{> 0}$ is called semi-convex if $\nu \ell(\nu)$ is convex; that is, $(\nu + 1) \ell(\nu + 1) - \nu \ell(\nu) \geq \nu \ell(\nu) - (\nu - 1) \ell(\nu - 1)$ for all $\nu \geq 2$. Further, $\mathsf{Poi}(\nu)$ denotes the Poisson distribution with parameter $\nu \in \N$.

When $\mathcal{L}_d$ contains all polynomials of degree $d$ with nonnegative coefficients, it follows that $\sup_{\ell \in \mathcal{L}} \rho_\ell$ is the \emph{$(d+1)$-Bell number}, which we denote by $\cB(d+1)$; it follows from existing asymptotic bounds that $\cB(d) \gtrsim d^{d^{1 - \epsilon}}$ for any $\epsilon > 0$~\citep{De81:Asymptotic}. By virtue of our $\PLS$ membership for both perfect (\Cref{theorem:polytope-perfect}) and proper equilibria (\Cref{thm:proper-pls}), we can now show the following.

\begin{theorem}
    \label{theorem:poa-lower}
    Consider the class of congestion games in which $\ell_r \in \mathcal{L}_d$ for every $r \in \cR$. If $\NP \neq \coNP$, we have $\mathsf{PoA}_{\mathsf{PerfEq}} \geq \cB(d+1) - \epsilon$ for any $\epsilon > 0$. Similarly, if $\NP \neq \coNP$, we have $\mathsf{PoA}_{\mathsf{PropEq}} \geq \cB(d+1) - \epsilon$ for any $\epsilon > 0$.
\end{theorem}

In fact, those lower bounds hold even for \emph{pure} perfect and \emph{pure} proper equilibria. \Cref{theorem:poa-lower} follows because in the contrary case, a problem in $\PLS$ would also be $\NP$-hard (by~\Cref{theorem:Gairing}), which is precluded subject to $\NP \neq \coNP$~\citep{Johnson88:How}. It is likely that the lower bound in~\Cref{theorem:poa-lower} can be pushed further through an explicit construction, as is the case for Nash equilibria; indeed, for Nash equilibria there is a certain gap between its price of anarchy and the best approximation ratio that can be achieved through a polynomial-time algorithm~\citep{Paccagnan24:Congestion}. We leave this question for future research.

\section{Equilibrium refinements for games in extensive form}\label{sec:efg}

This section introduces extensive-form games and equilibrium refinements therein.

\subsection{Notation and background}


Deciding which strategy each player should play in an extensive-form game of perfect recall can be encoded as a \emph{tree-form decision problem}. A tree-form decision problem for any player $i \in [n]$ is represented as a rooted tree, where each root-to-leaf path alternates between two types of nodes: \emph{decision points}, denoted as $j \in \cJ_i$ (which correspond to the infosets of the game), and \emph{observation points} (or \emph{sequences}), denoted as $\sigma \in \Sigma_i$. The root node of the tree, denoted as $\varnothing \in \Sigma_i$, is always an observation point. The edges of decision points are called  \emph{actions}, and the player must choose one of them before continuing. We denote the set of actions at decision point $j$ as $\cA_j$. At observation points, the labels are called \emph{observations}, one of which the player observes. We let $d_i = |\Sigma_i|$ denote the number of sequences in the tree.  The set of decision points that follow an observation point $\sigma \in \Sigma_i$, is denoted as $C_\sigma$. An observation point is denoted as $ja$, where $j$ is the the parent decision point and $j$ is the action that connects the two points. For any node $s$, we denote its parent node as $p_s$. Finally, we will use $\succeq$ to denote the partial order induced by the tree, \eg, $\Root \prec j$ for every decision point $j$.


A {\em sequence-form strategy}~\cite{Romanovskii62:Reduction,Koller94:Fast,Stengel96:Efficient} $\vx_i \in \R^{\Sigma_i}$ is a vector obeying the linear constraint system
\begin{align*}
    \vx_i[\varnothing] = 1, \;\;\; \vx_i[p_j] = \sum_{a \in A_j} \vx_i[ja] \;\; \forall j \in \cJ_i, \;\;\; \vx_i[\sigma] \geq 0 \;\; \forall \sigma \in \Sigma_i.
\end{align*}
We write $\cX_i$ to denote the set of all possible sequence form strategies for player $i$. Finally, for our purposes, an {\em extensive-form game} is a polytope game where each player's strategy set $\cX_i$ is a sequence-form strategy set.\footnote{Technically, this class of games is actually slightly more expressive than extensive-form games: extensive-form games would be equivalent to the class of games in which, additionally, utility functions are given explicitly as a (weighted) sum of monomials.}

\subsection{Computing extensive-form perfect equilibria}
\begin{definition} ($\eps$-perturbed EFG for extensive-form perfect equilibria)\label{def:qpe-perturbed-game}
    Let $\Gamma$ be an $n$-player extensive-form game and $\cX_i$ be the strategy set of player $i$. For some $\eps >0$, we let $\Gamma^{(\eps)}$ denote the perturbed sequence form game with the same structure as $\Gamma$ with the additional following restricted set of players' strategies. For player $i \in [n]$ we the perturbed sequence-form mixed strategies as 
    \begin{align*}
        \cX_i^{(\eps)} = \{\vx_i \in \cX_i: \vx_i[ja] \geq \eps \cdot \vx_i[p_j] \;\; \forall \sigma \in \Sigma_i\}.
    \end{align*}
\end{definition}
That is, in a perturbation for extensive-form perfect equilibria, every action $a$ is played at every decision point $j$ with probability at least $\eps$.

An extensive-form perfect equilibrium is now, as usual, defined as the limit of Nash equilibria of the above kind of $\eps$-perturbed games. The following can be easily shown by performing a bottom-up pass through the sequence-form tree:
\begin{theorem}\label{th:efpe}
    Best responses in $\cX_i^{(\eps)}$  can be computed in time polynomial in $d$ and the representation size of the utility vector. As a consequence, computing an EFPE of a potential game is in $\PLS$. 
\end{theorem}

In extensive-form games, one can also define notions of equilibrium refinements by viewing the game as a polytope game with strategy set $\cX_i$ for each player, and then applying the definitions in the main body for perfect and proper equilibria in polytope games. These give the notions, respectively, of normal-form perfect and normal-form proper equilibria. In extensive-form games, it turns out that extensive-form perfect equilibria are {\em incomparable} to both of these notions~\cite{Mertens95:Two}.

Similarly, a symbolic best-response for the perturbation that arises under QPEs can also be computed efficiently using the perturbation of~\citet{Gatti20:Characterization}, implying the following result.

\begin{theorem}\label{th:qpe}
    Computing a QPE of a potential game is in $\PLS$.
\end{theorem}

\section{Further omitted proofs}
\label{sec:furtherproofs}

Finally, we provide the proof of~\Cref{lemma:efficient-eval}, restated below.

\efficmult*

\begin{proof}
    We show the following claim inductively. For each gate $G$, let $L_G$ be the description length of the sub-DAG induced by $G$, including any rational constants or inputs in that sub-DAG. Let $P_G$ be the product of denominators of all rational constants and inputs in that sub-DAG. Clearly both $L_G$ and $\log_2 P_G$ are polynomial in the input length, and $\log_2 P_G \le L_G$. We claim that the output of every gate $G$ can be expressed as $v_G = a_G/P_G$, where $a_G$ is an integer and $|v_G| \le 2^{L_G}$. This would complete the proof immediately. We prove the claim inductively. The base case, for input or rational constant nodes, holds by definition. For the inductive case, let $g$ and $h$ be the two inputs of $G$.

    If $G$ is an addition gate, we have $|v_G|\le |v_g| + |v_h| \le 2^{L_g} + 2^{L_h} \le 2^{1 + \max\{L_g, L_h\}} \le 2^{L_G}$, and the denominator of $v_G$ is at most $P_G$ because the inputs to $G$ can be expressed with denominators $P_g$ and $P_h$ respectively, which by definition both divide $P_G$.

    If $G$ is a multiplication gate, then by definition we have $P_G = P_g P_h$ and $L_G > L_g + L_h$.  Therefore, we have $|v_G| = |v_g| |v_h| \le 2^{L_g} 2^{L_h} \le 2^{L_G}$ and $v_G = v_g v_h = a_g a_h / P_g P_h = a_g a_h / P_G$, as desired.
\end{proof}

\end{document}